\documentclass[AER,draftmode]{AEA}
\newcommand{\falsificationShare}{51\%}
\newcommand{\ncoShareOfFalsification}{75\%}
\newcommand{\nciShareOfFalsification}{24\%}
\newcommand{\otherShareOfFalsification}{21\%}
\newcommand{\poShareOfNco}{57\%}

\newcommand{\ivControlInPiTest}{24\%}

\newcommand{\ncCountOne}{35\%}

\usepackage{array}
\newcolumntype{P}[1]{>{\raggedright\arraybackslash}p{#1}}
\usepackage{appendix}
\usepackage[margin=1in]{geometry}  % avoid AEA conflict
\usepackage{graphicx}
\draftSpacing{1.2}

\usepackage{amsthm}
\usepackage{amsmath,bbm}

\usepackage{cancel}
\usepackage{enumitem}
\usepackage[hidelinks]{hyperref}
\usepackage{natbib}
\usepackage{soul}
\usepackage{amsfonts} 
\usepackage{silence}
\usepackage{caption}  % NB: in production, check caption compatibility with AER style
\usepackage[font=small]{subcaption} % multiple panels per figure
\usepackage{booktabs}
\usepackage{adjustbox} % for adding tables
\usepackage{etoolbox}  % for ifiselse in fig and tab commands
\usepackage{longtable} % multi-page appendix tables

\usepackage{pdflscape}
\usepackage{makecell}
\usepackage{mdframed}

\newcommand{\floatscale}{.95} % default scaling of floats

\newenvironment{addfig}[4][\floatscale]{%
\begin{figure}[htbp]
  \caption{#3}
  \centering
	\begin{minipage}{#1\linewidth} 
	  \includegraphics[width=\linewidth]{#2}
	\end{minipage}
	  \ifstrempty{#4}{% no figure notes
    }{% figure notes
    \medskip 
	\begin{minipage}{\floatscale\linewidth} % choose width
	{\footnotesize {#4}\par}
	  \end{minipage}
    }% notes
	\label{#2}
\end{figure} 
\clearpage % each exhibit in a page
}

\newcommand{\exhibitnotes}[1]{
\begin{minipage}{\textwidth}
\medskip
\footnotesize
\textit{Notes:} #1\par
\end{minipage}
}
\newenvironment{addtab}[4][\floatscale]{%
  \begin{table}[htbp]
  \caption{#3}
  \centering
	\begin{minipage}{#1\linewidth} % choose width suitably 
    \centering
    		\renewenvironment{table}[1][]{}{}
    		\renewcommand{\caption}[1]{}
    		\begin{adjustbox}{max width=#1\linewidth}
          \input{#2}
        \end{adjustbox}
  \end{minipage}
  \medskip
    \ifstrempty{#4}{% no notes
    }{% notes
    \smallskip
    
	\begin{minipage}{\floatscale\linewidth}
	{\footnotesize \textit{Notes:} {#4}\par}
	  \end{minipage}
    }% notes
  \label{#2}
  \end{table}
}

\renewcommand{\figurename}{}
\renewcommand{\thefigure}{Figure~\arabic{figure}}

\renewcommand{\tablename}{}
\renewcommand{\thetable}{Table~\arabic{table}}

\usepackage{tikz}
\usepackage{pgf}
\usetikzlibrary{positioning, calc, shapes.geometric, shapes.multipart, shapes, arrows.meta, arrows, 	decorations.markings, external, trees}
\tikzstyle{Arrow} = [
thick,
decoration={
	markings,
	mark=at position 0.99 withe {  % fix bug with tikz arrow head interference with decoration
		\arrow[thick]{latex}
	}
},
shorten >= 3pt, preaction = {decorate}
]
\tikzstyle{RedArrow} = [
red, 
line width = 0.5mm,
decoration={
	markings,
	mark=at position 0.99 with {
		\arrow[line width = 0.5mm,red]{latex}
	}
},
shorten >= 3pt, preaction = {decorate}
]

\tikzstyle{BlueArrow} = [
blue, 
line width = 0.5mm,
decoration={
	markings,
	mark=at position 0.99 with {
		\arrow[line width = 0.5mm,blue]{latex}
	}
},
shorten >= 3pt, preaction = {decorate}
]

\tikzstyle{DashedRedArrow} = [
thick,red,dashed,
decoration={
	markings,
	mark=at position 0.99 with {  % fix bug with tikz arrow head interference with decoration
		\arrow[thick,red]{latex}
	}
},
shorten >= 3pt, preaction = {decorate}
]

\newtheorem{thm}{Theorem}
\newtheorem{lem}{Lemma}
\newtheorem{dfn}{Definition}
\newtheorem{ass}{Assumption}

\newtheorem{corollary}{Corollary}

\usepackage[status=draft,inline,index]{fixme}
\fxsetup{theme=color,margin=false,mode=multiuser}
\FXRegisterAuthor{bw}{bwe}{BW}
\FXRegisterAuthor{dn}{dne}{DN}
\FXRegisterAuthor{od}{ode}{OD}
\FXRegisterAuthor{dz}{dze}{DZ}
\FXRegisterAuthor{iw}{iwe}{IW}

\DeclareMathOperator*{\argmin}{arg\,min}
\def \bbE {\mathbb{E}}\newcommand\indep{\protect\mathpalette{\protect\independenT}{\perp}}
\def\independenT#1#2{\mathrel{\rlap{$#1#2$}\mkern2mu{#1#2}}}
\newcommand{\beginQuotable}[1][quote name]{\phantomsection\label{Q:#1}}
\def\endQuotable{}

\title{Negative Control Falsification Tests for\\ Instrumental Variable Designs\thanks{We thank Maya Orenstein for excellent research assistance. We thank Kirill Borusyak, Yoav Goldstein, Matan Kolerman,  David Lee, Ro'ee Levy, Ashesh Rambachan, Tamir Zehavi, and seminar participants at the ASSA 2023 Meeting, EUROCIM 2024, Princeton University, Stanford University, Technion, Hebrew University of Jerusalem, Tel Aviv University, EIEF Conference and SOLE 2024 meeting for their helpful comments and suggestions. We acknowledge support from the Tel Aviv University Center for AI and Data Science (TAD) in collaboration with Google as part of the initiative of AI and DS for social good and from the Sapir Center for Development.}}    

\shortTitle{Negative Controls for IV Designs}

\author{Oren Danieli, Daniel Nevo, Itai Walk, Bar~Weinstein, Dan~Zeltzer\thanks{All authors are from Tel Aviv University. Danieli is the corresponding author. Email: orendanieli@tauex.tau.ac.il.}}

\date{\today}
\pubMonth{}
\pubYear{}
\pubVolume{}
\pubIssue{}

\begin{abstract}
The validity of instrumental variable (IV) designs is typically tested using two types of falsification tests. We characterize these tests as conditional independence tests between negative control variables---proxies for unobserved variables posing a threat to the identification---and the IV or the outcome. We describe the conditions that variables must satisfy in order to serve as negative controls. We show that these falsification tests examine not only independence and the exclusion restriction, but also functional form assumptions. 
Our analysis reveals that conventional applications of these tests may flag problems even in valid IV designs.  We offer implementation guidance to address these issues.
\end{abstract}
\begin{document}

\maketitle

% consider Miguel, Satyanath, and Sergenti, 2004 for exclusion

The identification assumptions in instrumental variable (IV) designs cannot be directly tested. Instead, researchers often use indirect falsification, or ``placebo,'' tests. 
Reviewing the most-cited papers in five leading economics journals, we find that \falsificationShare~of IV studies employ such falsification tests.
The large majority of falsification tests fall into two categories: \ncoShareOfFalsification~of papers that implemented a falsification test examined that the IV is not associated with certain variables such as lagged outcomes, which we call \emph{negative control outcomes}, (NCOs). Similarly,   \nciShareOfFalsification~examined that the outcome is not associated with other variables, which we call \emph{negative control instruments} (NCIs).
For example, they tested that the outcome is not correlated with variables that resemble the IV but do not affect the treatment. \beginQuotable[CiteLitAtStart]Extensive literature has developed a theoretical framework for negative control falsification tests in classic causal settings \citep{lipsitch2010negative,shi2020selective}, but not for the assumptions underlying IV designs.\footnote{In epidemiology and biomedical fields, researchers use the similar terminology of negative controls for falsification tests to detect potential confounding in an exposure-outcome relationship.}\endQuotable{} 
This paper aims to fill this gap. 

We propose that practitioners first identify potential threats to their IV validity, which we formally define. Since variables posing these threats are typically unobserved, researchers should look for proxy variables, termed \textit{negative controls}. We characterize the conditions that these proxies must satisfy and show how to use them to test the validity of the IV design using (conditional) independence tests.
The theory highlights two common pitfalls, which can falsely flag problems in valid IV designs. First, NCI tests typically require conditioning on the IV---a step frequently overlooked in practice. Second, prevalent negative control tests may flag violations of unnecessary or replaceable functional form assumptions. We propose ways to separately test the validity of the IV design from these functional form assumptions. Our framework also suggests novel, underutilized negative control variables and testing methods. 

\beginQuotable[APVIntro]
We first introduce the concept of an \emph{alternative path variable}, which represents a variable that poses a threat to the identification. In valid IV designs, which satisfy independence and the exclusion restriction, the only path between the IV and the outcome is through the treatment.\footnote{The independence assumption typically includes independence of the IV with the potential treatment and the potential outcome \citep{abadie2003semiparametric}. The falsification tests we discuss in this paper focus only on outcome independence. See Section \ref{subsec:threat}.} Threats to identification can be characterized as alternative paths between the IV and the outcome through an alternative path variable, rather than through the treatment.
\endQuotable{}

To construct a negative control test, researchers need to determine which type of alternative path threatens the IV design. We distinguish between two categories of variables that can create such paths. In the first category, alternative path outcome (APO) variables, the concern is that a variable that is associated with the outcome would also be associated with the IV. \ref{fig:dags_intro} illustrates two examples of such cases.\footnote{Throughout the paper, we use directed acyclic graphs (DAGs) to visualize complex structures, as advocated by \cite{imbens2020potential}. \ref{AppSec:DAG} outlines the theory presented in this paper within the formal causal DAG framework \citep{pearl2009causality}} Panel~A shows a potential violation of the independence assumption. For concreteness, consider the context of \cite{martin2017bias}, who examine the impact of Fox News viewership ($X$) on Republican vote shares ($Y$). As an IV, they use the local Fox News cable channel position ($Z$) since lower channel numbers induce higher viewership.
One concern is that unobserved local conservativeness ($U$, the APO variable) affects not only the republication vote share (the outcome) but also the channel position (the IV), as marked with the dashed arrow. If that is the case, an alternative path between the channel position and Republican voting share exists via local conservativeness, the APO variable. This would violate the independence assumption and invalidate the design.
Panel~B offers an example of an APO variable ($U_2$) that is part of a potential violation of the exclusion restriction assumption. 

In the second category, alternative path instrument (API) variables, the concern is that variables known to be associated with the IV would also be associated with the outcome. Panel~C  of \ref{fig:dags_intro} provides an example of an API variable that potentially violates the independence assumption. For concreteness, consider the context of \cite{nunn2014us}, who examine the effect of US food aid ($X$) on conflicts in recipient countries ($Y$). They use the US production of wheat ($Z$), a staple aid crop, as an IV for aid. Here, the API variable is unobserved weather conditions ($U_3$).  It is known that weather conditions affect wheat production, the IV. The question is whether they also affect conflicts, the outcome, as marked with the dashed arrow. If so, an alternative path from wheat production to conflict exists via the API variable. 
Panel~D offers an example of an API variable ($U_4$) that is part of a potential violation of the exclusion restriction assumption. 

\beginQuotable[NCIntro]
Since alternative path variables of both categories are often unobserved, researchers can utilize negative control variables as proxies for them. A negative control outcome (NCO) is a proxy for an APO variable. An association between an NCO and the IV implies the presence of an alternative path, indicating that the design is not valid. In \citeauthor{martin2017bias}, the lagged outcome, Republican vote share in 1996 ($NC_1$ in \ref{fig:dags_intro}) is used as an NCO for unobserved conservativeness ($U_1$). If 1996 vote shares correlate with channel position, it would imply that cable companies consider the population conservativeness when placing the channels, such that the dashed arrow exists. Hence, there is an alternative path between the IV and the outcome, violating the independence assumption. Panel~A of \ref{tab_intro_examples} describes further examples of applications of NCO variables in economic research. 

Similarly, a negative control instrument (NCI) can be used as a proxy for an API variable. For example, \citeauthor{nunn2014us} use orange production as an NCI ($NC_3$ in \ref{fig:dags_intro}) for unobserved weather conditions ($U_3$)---orange production is affected by similar weather conditions as wheat but is not used for food aid. If orange production were associated with conflicts conditional on wheat production, this would indicate an alternative path between the IV and the outcome, violating the independence assumption.
Panel B of \ref{tab_intro_examples} lists additional examples of NCI variables in economic research. \endQuotable{}

The definition of negative controls clarifies which proxy variables researchers can use as NCOs or NCIs. This definition guarantees that these proxies can test for alternative paths using (conditional) independence tests. In particular, variables directly associated with the IV (not through the APO variable) cannot serve as NCOs. 
Similarly, variables associated directly with the outcome (not through the API variable or the IV) cannot serve as NCIs. Tests using such variables would falsely flag valid IV designs.

The theory highlights two common pitfalls in current practice. First, in our survey, the vast majority of papers using NCI variables implemented a test that will falsely flag problems in valid IV designs (with sufficient sample size). In almost all these papers, the NCI was a variable that is similar to the IV but does not affect the treatment (such as orange instead of wheat production in the previous example). Researchers typically test whether such an NCI is correlated with the outcome by plugging it instead of the original IV in the reduced form equation. This specification overlooks a key issue: NCIs are typically correlated with the original IV and therefore will be associated with the outcome due to this correlation, even in a valid IV design. For example, as shown in Panel C of \ref{fig:dags_intro}, both orange production ($NC_3$) and wheat production ($Z$) are influenced by weather conditions ($U_3$). This means that orange production and conflict ($Y$) would be correlated even if wheat production is a valid IV (i.e., the dashed arrow does not exist and there is no alternative path).
We show that this problem can be avoided by controlling for the original IV in the NCI test. In the above example, this means controlling for wheat production when testing the association between orange production and conflict. 

The second pitfall is that negative control tests may flag problems even in valid IV designs due to a misspecified functional form. 
In 2SLS specifications, researchers must choose a functional form for the IV and the control variables, typically assuming a simple linear-additive model. This structure often carries over into the execution of negative control tests. 
Consequently, even valid IV designs may fail negative control tests merely due to violations of functional form assumptions.
However, unlike the independence and exclusion assumptions, which are necessary for IV validity, functional form assumptions can often be relaxed and, in some cases, are unnecessary for identifying causal effects. 
To address this problem, researchers can use alternative negative control tests that rely on weaker functional form assumptions.  

The theory can also be used to identify new types of negative control variables, some of which have yet to be commonly employed in empirical research. For example, variables that causally influence the IV could serve as NCIs (e.g., observed weather conditions can serve as NCIs when the IV is wheat production).

\beginQuotable[FirstLitReview]
This paper adds to prior econometrics work on tests for IV design validity and, more generally, the validity of causal designs. Recent work has suggested novel tests to examine the validity of IV designs \citep{kitagawa2015test, huber2015testing, mourifie2017testing, frandsen2023judging, chyn2024examiner}. Previous work has also discussed robustness tests, not specific to IV, based on varying the set of controls \citep{altonji2005selection, oster2019unobservable, diegert2022assessing}. However, \citet{pei2019poorly} recommends using such control variables as NCOs instead. \cite{eggers2021placebo} discuss the usage of placebo tests in the social sciences more broadly. We contribute to this literature by outlining a theoretical framework for the most common type of falsification tests for IV designs.
\endQuotable{}

\beginQuotable[LitReview]
This paper also contributes to the growing literature on negative controls \citep{lipsitch2010negative,shi2020selective}. In a standard causal design, negative controls are used to detect or even correct for bias due to unobserved treatment-outcome confounding. To this end, valid and, in some cases, even invalid IVs can 
serve as \textit{negative control exposures}, 
and under further assumptions can be combined with an additional negative control to achieve point identification 
\citep{miao2018identifying, shi2020selective, tchetgen2024introduction, dukes2024using}.
We apply the theory of negative controls in a different setting, where researchers employ an IV design and seek to use negative controls specifically for testing the IV assumptions in order to assess the design's validity. 
Our work is related to \cite{davies2017compare}, who use negative controls in IV contexts without developing a theoretical framework for such an approach. \endQuotable{}We find several important differences in the theoretical framework of negative controls for assessing IVs compared to their usage in the standard treatment-outcome setting.

The rest of this paper proceeds as follows. Section~\ref{sec:survey} surveys the current practice of falsification tests for IV designs. Section~\ref{Sec:theory} presents the theory for negative control tests in IV designs. Section~\ref{sec:practice} provides guidance for practitioners and demonstrates key findings using recent empirical studies. Section~\ref{sec:conclusion} concludes.
 
\section{Survey of Current Practice}
\label{sec:survey}
To provide an overview of current practices in falsification testing for IV designs, we surveyed the most highly cited articles with an IV analysis published between 2013 and 2023 in top economics journals. We then classified the characteristics of the falsification tests used. \ref{AppSec:survey} provides additional details on the survey construction, and the results are summarized in \ref{tab_survey}.

We highlight five key findings from this survey. 
First, falsification tests are widely used in IV analyses. Approximately half (\falsificationShare) of all articles surveyed employ some form of falsification test (Column~2 of \ref{tab_survey}).

Second, most falsification tests fall within the negative control framework described in the introduction and formalized in this paper. As outlined earlier, these tests can be divided into two types: \emph{negative control outcome} (NCO) tests, which check for associations between the IV and variables it should not be associated with, and \emph{negative control instrument} (NCI) tests, which examine associations between the outcome and variables it should not be associated with. Among surveyed papers using falsification tests, \ncoShareOfFalsification~used NCO tests (Column~3) and~\nciShareOfFalsification~used NCI tests (Column~4). All other types of falsification tests combined were used in \otherShareOfFalsification~of the papers (Column~5); \ref{AppSec:survey} lists these other, less common types. 

Third, current applied work usually restricts itself to two simple types of negative control test specifications that rely on the 2SLS functional form assumptions. NCO tests typically involve estimating a revised reduced form equation using an alternative outcome (e.g., a lagged outcome) and testing if it is unrelated to the IV (Column 6). Such specifications account for \poShareOfNco ~of all NCO tests. The remaining NCO tests often follow a similar logic.\footnote{For example, balance tables that regress various NCOs on the IV.}  
For NCI tests, researchers either replace the original IV in the reduced form equation with a similar variable that does not affect the treatment (Column~7), or add this variable to the reduced form equation (Column~8). %\replaceOrAddShareOfNCI surveyed papers using NCI tests used these specifications. 

Fourth, most reported NCI tests are implemented incorrectly. As noted in the introduction and discussed in detail later, NCI tests should almost always control for the original IV. In practice, only \ivControlInPiTest ~of the papers surveyed reported doing so. With sufficient sample size, this error leads to finding false problems in valid IV designs. It is likely that additional NCI tests were conducted incorrectly, finding false problems in valid designs, and were therefore not reported.

Finally, papers using falsification tests usually utilize only a few negative control variables. The median number of negative control variables used in the surveyed papers is 3.5 (Column~9), with \ncCountOne~of of the papers using only one. This finding suggests that researchers use only a subset of the available relevant negative controls. As we demonstrate in Section~\ref{sec:practice}, the theory can guide a systematic search for negative control variables in existing data and suggest novel types of negative control variables researchers can use to evaluate their IV designs.

\section{Theory of Negative Controls in IV Settings}
\label{Sec:theory}

%preamble to clarify that important things come at the end
In this section, we present the theory of negative control tests for IV designs. 
%We start with a classification of the threats to the design. 
%Subsequently, we introduce negative controls as proxies for the variables posing these threats and discuss the assumptions negative controls must satisfy. Then, we present negative control falsification tests as conditional independence tests that use these proxies. We show that NCI tests typically require conditioning on the IV. Finally, we discuss the inclusion of control variables and show that negative control tests can depend on functional form assumptions that are modifiable and sometimes unnecessary. 
The theory is constructed using the terminology of potential outcomes, and we use DAGs for examples and intuition. \ref{AppSec:DAG} introduces the basic relevant concepts for DAG theory and replicates the key definitions and theorems using DAGs. 

\beginQuotable[threatSection]
\subsection{The Threats to the Identification}\label{subsec:threat}
\subsubsection{IV Assumptions}
\label{SubSec:setup}
Consider i.i.d. units indexed by $i=1,\ldots,n$. Denote the observed (endogenous) treatment status by $X_i$, and the candidate IV by $Z_i$. Let $Y_i(z,x)$ be the potential outcome for unit $i$ had $Z_i$ and $X_i$ been jointly set to the values $z$ and $x$, respectively.\footnote{This formulation implicitly assumes the stable unit treatment value assumption (SUTVA).}
We make the standard assumption that the observed outcome $Y_i$ is given by $Y_i=Y_i(Z_i,X_i)$. Because units are assumed to be i.i.d., we omit the subscript $i$ when it improves clarity. All variables may be discrete or continuous.

The negative control tests we discuss in this paper examine whether there is an alternative path between the IV and the outcome, in addition to the standard path through the treatment. Such an additional path would violate one of the following two assumptions.
The first assumption, \emph{outcome independence}, maintains that IV assignment is independent of the potential outcomes.
\begin{ass}[Outcome independence] For all $z$ and $x$,
\label{Ass:Indep}
$Z \indep Y(z,x).$
\end{ass} 
\noindent This assumption is usually written as part of a more general independence assumption \citep[e.g.,][]{abadie2003semiparametric}. Here, we distinguish between outcome independence and \emph{treatment independence}, which requires $Z \indep X(z)$ for every value of $z$. 
Only outcome independence is tested in the negative control tests we discuss in this paper.

Outcome independence is violated if the IV is affected by a variable that also affects the outcome. As previously discussed, \citet{martin2017bias} study the effect of Fox News on voting using channel positions as an IV. This example is illustrated in Panel A of \ref{fig:dags_intro}. The concern is that cable companies accounted for local conservativeness ($U_1$) when assigning Fox News channel position ($Z$), as illustrated by the dashed arrow. If so, an alternative path emerges between channel position and voting ($Y$) through local conservativeness. This path violates outcome independence. 
%Panel C of \ref{fig:dags_intro} illustrates another potential violation of outcome independence.

The second assumption, \emph{exclusion restriction}, maintains that the IV does not have a direct effect on the outcome.
\beginQuotable[Yx]Let $Y(x)$ be the potential outcome had the treatment $X$ been set to $x$, while $Z$ had not been set to any particular value and takes its natural value, i.e., $Y(x)=Y(Z,x)$.  
\endQuotable{}
\begin{ass}[Exclusion restriction] For all $z,x$,
\label{Ass:ExclRestriction}
$Y(z,x)= Y(x).$
\end{ass} 
\noindent The exclusion restriction is violated if the IV affects the outcome in alternative ways, in addition to its effect through the treatment. 

\beginQuotable[exclusionExample]
Panel B of \ref{fig:dags_intro} illustrates a potential violation of the exclusion restriction assumption. For example, \citet{angrist1996children} study the effect of the number of children ($X$) on female labor supply ($Y$), using the sex composition of the first two children as an IV ($Z$) (because same-sex children induce further births for parents who have a preference for gender variety). One potential concern is that same-sex sibship could reduce housing expenditures ($U_2$) due to hand-me-downs, which could then affect female labor supply decisions. In this case, an alternative path between the sex composition of the first two children and female labor supply would exist through the effect on household expenditures. This violates the exclusion restriction.  
\endQuotable{}

Together, outcome independence and exclusion restriction imply\footnote{\beginQuotable[YxNoExclusion]When the exclusion restriction does not hold, $Y(x)$ is still properly defined but does not equal $Y(z,x)$ for every value of $z$. Therefore,  recalling that $Y(x)=Y(Z,x)$, violation of the exclusion restriction implies $Z\cancel{\indep}Y(x)$.\endQuotable{}} 
\begin{equation}
\label{Eq:IVexogeneity}
  Z \indep Y(x) \text{ for all } x. 
\end{equation}
In loose terms, \eqref{Eq:IVexogeneity} requires that there are no alternative paths between the IV and the outcome except through the treatment. Because potential outcomes are never observed, neither these two assumptions nor \eqref{Eq:IVexogeneity} can be tested directly.

To identify a causal effect using an IV design, additional assumptions are also necessary. For example, the design of \citet{angrist1996identification} also requires treatment independence, relevance, and monotonicity. However, the negative control tests presented in this paper do not test these other assumptions.

%\subsection{Negative Control Outcomes}
%\label{SubSec:NCOtheory}
%\subsubsection{Alternative Path Outcome Variables}
%\label{SubSubSec:APOvar}

%We characterize negative controls as proxies for threats to IV validity.  

\subsubsection{Alternative Path Variables}

To formalize the notion of a threat to the identification, we introduce the concept of an \emph{alternative path variable}---a variable that is part of a suspected alternative path between the IV and the outcome that, if such a path exists, would violate outcome independence or exclusion.
For simplicity, we assume that only one potential threat to the IV validity exists. \ref{AppSec:proofs} addresses a more general case with multiple threats. We distinguish between two types of alternative path variables that require different types of falsification tests.

\beginQuotable[APO]The first type of identification threat involves \textit{alternative path outcome} (APO) variables.
These variables are presumably associated with the outcome. The threat is that they are also associated with the IV, which would generate an alternative path. In \cite{martin2017bias}, conservativeness of the local population is an APO variable ($U_1$ in Panel A of \ref{fig:dags_intro}). Since conservativeness certainly affects voting behavior, it poses a threat if it also affects cable companies' decisions for channel position (represented by the dashed arrow).\endQuotable{} Panel B of \ref{fig:dags_intro} describes an APO variable for a potential violation of the exclusion restriction.

The formal definition of an APO variable is as follows.
\begin{dfn}[Alternative path outcome variable\label{dfn:apo_var1}]
A random variable $U$ is an APO variable if the following two conditions hold.
 \begin{enumerate}
\item  \textit{Latent IV validity}. $Z \indep Y(x) | U$.
\item \textit{Path indication}. If $Z \indep Y(x)$ then $Z \indep U$.
\end{enumerate}
\end{dfn}
\noindent{}  

Latent IV validity posits that had we observed and conditioned on the APO variable, the IV design would have been valid---both outcome independence and the exclusion restriction would hold conditional on the APO variable.
This condition implies that imperfect proxies for variables posing identification threats cannot themselves be APO variables, as controlling for an imperfect proxy does not make the IV and the potential outcome conditionally independent. Using again the example of \cite{martin2017bias}, the share of Republican votes in 1996 is only an imperfect proxy for the APO variable (latent conservativeness), and hence controlling for it does not eliminate the threat. Therefore, the Republican vote share in 1996 is not an APO variable as it does not satisfy the latent IV validity condition. 
\beginQuotable[LatentEx]Latent IV validity is analogous to the latent exchangeability assumption appearing in recent literature on negative controls in epidemiology \citep{shi2020selective} and statistics \citep{tchetgen2024introduction}.\endQuotable{}

Path indication states that a valid IV is not associated with the APO variable. Its contrapositive ensures that an association between the IV and the APO variable implies an alternative path between the IV and the potential outcome. 
Path indication guarantees that if there is a path from the IV to the APO variable, the path continues from the APO variable to the outcome. 
Therefore, it excludes variables unrelated to the outcome, as they can be associated with the IV without implying anything about the design validity. 
While APO variables often causally affect the outcome, it is not mandatory (as demonstrated in \ref{AppSubSec:NonCausalAPV}). Path indication also rules out variables that could be related to both the IV and the outcome without generating a correlation between them. For example, this would occur if a variable is correlated with the outcome for some subpopulation but potentially correlated with the IV only for a separate subpopulation. 
Two examples are provided in \ref{AppSubSec:ViolationPathHetero} and \ref{AppSubSec:ViolationPathMultiDim}. 

The second type of alternative path variables is \emph{alternative path instrument} (API) variables. 
API variables are known to be associated with the IV, and the concern is an association they might have with the outcome. This is in contrast to APO variables, which are known to be associated with the outcome, and the concern is their possible association with the IV. U sing the previous example of \citet{nunn2014us}, illustrated in Panel~C of \ref{fig:dags_intro}, 
weather conditions ($U_3$) is an API variable. Wheat production ($Z$) is known to be affected by weather. An alternative path that threatens identification may form if weather also affects conflicts ($Y$) directly (the dashed arrow exists).
Panel~D describes an API variable for a potential violation of the exclusion restriction.

Formally, an API variable satisfies the following definition. 
\begin{dfn}[Alternative path instrument variable\label{dfn:api_var1}]
A random variable $U$ is an API variable if the following two conditions hold.
\begin{enumerate}
\item  \textit{Latent IV validity}. $Z \indep Y(x) | U$.
\item \textit{Path indication}. If $Z \indep Y(x)$ then $U \indep Y|Z$.
\end{enumerate}
\end{dfn}
\noindent{}This definition resembles the definition of APO variables (Definition \ref{dfn:apo_var1}). The first condition, latent IV validity, is exactly as before. 
The difference between API and APO variables is encapsulated in the second condition, path indication. For API variables, this condition requires that if $Z \indep Y(x)$, then the API variable must be independent of the observed outcome conditional on the IV (i.e., $U \indep Y|Z$). Through its contrapositive, this condition implies that an association between the API variable and the outcome, not via the IV, indicates that there exists an alternative path between the IV and the outcome. Therefore, the IV is invalid. Typically, path indication is satisfied when the API variable is associated with the IV. 

For API variables, path indication rules out variables that are associated with the outcome through the treatment (conditional on the IV). 
Such variables are not informative about the validity of the IV design as they are associated with the outcome through the treatment, even if the IV design is valid. This is different from APO variables that could be associated with the treatment (even conditional on the IV). See \ref{AppSubSec:APvarAndTreatmnet} for an example and a further discussion of this issue. This implies that APO variables can be associated with (or even identical to) the original confounder of the treatment-outcome relationship (which we mark by $W$ in our examples). However, an API variable cannot. 
\endQuotable{}

\subsection{Negative Control Variables}
\label{SubSubSec:NCassumptions}

\beginQuotable[NCareProxy]
Negative control variables are observed proxies for the unobserved alternative path variables.
Building on the definitions of alternative path variables, we are now ready to formalize the assumptions required for a random variable to serve as a negative control. The first type of negative control, \emph{negative control outcome}, is a proxy for an APO variable. Observed variables can serve as NCOs if they satisfy the following definition.
\endQuotable{}

\beginQuotable[NCODefinition] 
\begin{dfn}
[Negative control outcome]
\label{dfn:NCO_var}
A random variable $NC$ is an NCO if there exists an APO variable $U$ such that the following two conditions hold.
\begin{enumerate}
    \item The NCO assumption.  $NC\indep Z|U$.
%\begin{equation*}
 % Z\indep NC|U.  %\label{nc_assumption}
%\end{equation*}
 \item $U$-comparability. $NC \cancel{\indep} U$.
\end{enumerate}
\end{dfn}
\endQuotable{} 

\beginQuotable[PathNotAssociation]
The NCO assumption guarantees that any path between the IV and the NCO must go through the APO variable $U$. It rules out variables that have other paths to the IV.\endQuotable{} 
\beginQuotable[NCOExample] Panel A of \ref{fig:dags_intro} demonstrates the NCO assumption in a setting with a potential violation of outcome independence. 
For example, consider again the effect of Fox News on voting \citep{martin2017bias}. The lagged outcome---Republican vote share in 1996---is an NCO ($NC_1$). The NCO assumption requires that any association between lagged voting and later channel position assignment ($Z$) arises only due to local conservativeness ($U_1$, the APO).\endQuotable{}
Panel~B of \ref{fig:dags_intro} depicts an example of an NCO ($NC_2$) that satisfies this assumption where a violation of the exclusion restriction is the concern. 

The NCO assumption is violated for variables that are directly related to the IV, not through an APO variable. This can occur if the IV affects the candidate for NCO, either directly or through the treatment or the outcome. 
For example, various IV studies on the impacts of exposure to air pollution on different outcomes use non-respiratory hospital admissions, a seemingly unrelated outcome, as NCOs. One might expect that these admissions would only correlate with flawed IVs for air pollution. However, \cite{guidetti2021placebo} demonstrate otherwise. They find that air pollution increases non-respiratory admissions through hospital congestion caused by a surge in respiratory admissions. Therefore, non-respiratory admissions are not informative about the IV validity as they correlate with both flawed and valid IVs. Formally, non-respiratory admissions correlate with the IV, not through any APO variable but due to the unrelated mechanism of congestion. Hence, non-respiratory admissions violate the NCO assumption. 
%\ref{AppSubSec:nc_dags_violation} shows two additional theoretical examples of such violations. 
 
\beginQuotable[NCOUcomparability]
The $U$-comparability assumption guarantees that the NCO has a path to the APO variable. This assumption guarantees that the NCO is a relevant proxy for the APO variable.
For example, voting in 1996 satisfied $U$-comparability as it is correlated with the APO, unobserved conservativeness in the region. 
This assumption rules out variables that are uninformative about the design validity because they are unrelated to the identification threat (and specifically to the APO variable).
\endQuotable

\beginQuotable[NCIisProxy]
The second type of negative control variable, \emph{negative control instrument}, is a proxy for API variables. \endQuotable{}Observed variables can serve as NCIs if they satisfy the following definition. 
\begin{dfn}[Negative control instrument]
\label{dfn:NCI_var}
A random variable $NC$ is an  NCI if there exists an API variable $U$ such that the following two conditions hold.  \begin{enumerate}
    \item The NCI assumption. $NC \indep Y|Z,U$.
    \item $U$-comparability $NC \cancel{\indep} U|Z$.
\end{enumerate}
\end{dfn}
\noindent \beginQuotable[PathsNCIAssumption] The NCI assumption guarantees that any path between the outcome and the NCI goes through the API variable $U$ or the IV $Z$. It rules out variables that have other paths to the outcome. 
\endQuotable{}

While similar, the NCI assumption and the NCO assumption (Definition \ref{dfn:NCO_var}) differ in three key aspects. First, the alternative path variable $U$ is an API variable instead of an APO variable. Second, the conditional independence is between the NCI and the outcome instead of the IV. Due to these two differences, the NCI tests defined below test for a potential association with the outcome and not with the IV. The third difference is that the independence requirement is also conditional on the IV. This is because in valid IV designs, the NCI is often associated with the outcome through the IV, as we discuss in the next section. 

Panel C of \ref{fig:dags_intro} demonstrates the NCI assumption in a setting with a potential violation of outcome independence. For example, consider again the context of \citet{nunn2014us}, which uses an alternative crop (e.g., oranges) production as an NCI ($NC_3$). Orange production is affected by similar weather conditions ($U_3$) as wheat production ($Z$) and, therefore, would be correlated with it. However, unlike wheat, oranges are not used as food aid ($X$). Therefore, orange production is unrelated to conflicts ($Y$), conditional on both weather and wheat production.    

In many applications, the NCI assumption rules out a large class of observed variables because of their association with the outcome. For example, demographic variables often exhibit an association with the outcome, even conditional on the IV and the API variable, and therefore cannot serve as NCIs. Variables that are associated with the treatment are also not NCIs as they are also associated with the outcome conditional on the IV and the API variable. Moreover, in cases where the IV effect on the outcome is heterogeneous, any variable associated with the source of heterogeneity cannot serve as an NCI. 
In practice, the NCI assumption is more restrictive than the NCO assumption. The reason is that the NCO assumption requires conditional independence with the IV, which is typically more plausible than conditional independence with the outcome. 

\beginQuotable[NCIUcomparability]
The $U$-comparability assumption for NCIs implies that the NCI is indeed a proxy for the API variable. In contrast to $U$-comparability for NCOs, for NCIs, their association with the API variable must exist conditionally on the IV. This assumption rules out variables that are not informative about the IV validity because they are unrelated to the identification threat (and specifically to the API variable), conditional on the IV.
\endQuotable

\beginQuotable[AnalogytoTT]
The NCO and NCI definitions are analogous to the conditions that were formalized in previous literature on negative controls. In particular, U-comparability is common in the literature on negative controls  \citep{lipsitch2010negative,shi2020selective}. The NCO and NCI assumptions are similar to the conditional independence assumption of \citet{tchetgen2024introduction} (see equations 12 and 13).
\endQuotable{}

\beginQuotable[APVasNC]Definitions \ref{dfn:NCO_var} and \ref{dfn:NCI_var} imply that alternative path variables are themselves negative controls. This is because APO and API variables trivially satisfy both conditions in the definitions. For example, in the previously discussed design of \citet{angrist1996children}, the concern is that the sex composition of the first two children (the IV) may influence household expenditures (the APO variable) due to hand-me-downs, forming an alternative path to female labor supply (the outcome). \citet{rosenzweig2000natural} explore this by using a dataset in which clothing expenditures are observed, and use this APO variable as an NCO.\endQuotable{} 

The NCO and NCI assumptions can be weakened to cover more variables that are informative about the validity of the IV design. In \ref{AppSec:proofs}, we offer a more general definition of negative controls that allows for direct associations between the NCO and the IV or the NCI and the outcome, not through the alternative path variable if the design is invalid.

\subsection{Negative Control Tests}
\label{SubSubSec:NCOtest}

A \emph{negative control outcome test} (NCO test) is any statistical test of independence between the IV and an NCO. The null hypothesis is $H_0$: $Z \indep NC$. For example, \citet{martin2017bias} regress their NCO, Republican vote share in 1996, on the IV, Fox channel positioning. Under the null, the coefficient on the IV in this regression should equal zero. Indeed, they found no evidence to reject this hypothesis, which supports their design validity.

The following theorem states that rejecting the null hypothesis implies a violation of outcome independence or the exclusion restriction.

\begin{thm}%[Negative Control Outcome Test]
\label{Thm:NCOdep}
Assume that a random variable $NC$ is an NCO (Definition~\ref{dfn:NCO_var}). If $NC\cancel{\indep}Z$, then either outcome independence or exclusion restriction is violated. That is, the IV design is invalid.
\end{thm}

All proofs are given in \ref{AppSec:proofs}. The appendix proof covers a more general version of Theorem~\ref{Thm:NCOdep} for designs that include control variables (discussed in Section \ref{SubSubSec:FuncForm}). For the case without controls, the sketch of the proof is as follows. By the NCO assumption, the dependence between the IV and the NCO implies an association between the IV and an APO variable ($Z \cancel{\indep} U$). By path indication, $Z \cancel{\indep} U$ indicates an alternative path between the IV and the outcome ($Z\cancel{\indep}Y(x)$); i.e., the IV design is invalid. 

Similarly, a \emph{negative control instrument test} (NCI test) examines whether the outcome and the NCI are independent, conditional on the IV. Formally, the statistical test is for the null hypothesis $H_0:$ $NC \indep Y|Z$.
If the NCI is associated with the outcome conditional on the IV, this necessarily implies that the IV design is not valid, as stated in the following theorem.
\begin{thm}%[Negative Control Instrument Test]
\label{Thm:NCIdep}
Assume that a random variable $NC$ is an NCI (Definition \ref{dfn:NCI_var}). If $NC\cancel{\indep}Y|Z$, then either outcome independence or exclusion restriction is violated. That is, the IV design is invalid.
\end{thm}

For example, \citet{nunn2014us} regress their outcome, conflicts, on various alternative crop production (e.g., oranges), which are the NCIs, controlling for the original IV, wheat production. They are unable to reject a zero coefficient on alternative crops. Hence, they do not find an indication of a problem with the IV. 

NCI tests typically require conditioning on the IV, as the NCI may be associated with the outcome even in valid IV designs. This association arises because the NCI is often associated with the IV, which in turn influences the outcome through the treatment. For example, in Panels C and D of \ref{fig:dags_intro}, the NCI and the outcome are associated through the IV, even if no alternative path exists and the IV design is valid. In \cite{nunn2014us}, orange production (the NCI) is associated with conflicts (the outcome), as both are associated with wheat production (the IV).

However, if the NCI and IV are independent, conditioning on the IV is not required. In such cases, researchers can use an unconditional independence test for the null $H_0:$ $NC \indep Y$, as formalized in the following theorem.

\begin{thm}
\label{thm:NCIdepUnconditional}
Assume that a random variable $NC$ satisfies the NCI assumption. If in addition $NC \indep Z$, then if  $NC\cancel{\indep}Y$, either outcome independence or exclusion restriction is violated. That is, the IV design is invalid.
\end{thm}

\beginQuotable[T3Remark]
\noindent \ref{AppSec:proofs} provides a version of this theorem with control variables, in which the IV and NCI need to be conditionally independent only given the set of controls. 
\endQuotable{}

\beginQuotable[T3Example]Situations where $NC\indep Z$ (so, per the theorem, unconditional NCI tests may be valid) can occur when considering violations of the exclusion restriction assumption. Panel~A of \ref{fig:scenarios} provides an example. Consider the context of \citet{jacob2007crime}, who study the effect of lagged crime ($X$) on current crime ($Y$). They use lagged weather as an IV ($Z$) for lagged crime. The API variable is temporal displacement of economic activity ($U$): Lagged weather can postpone economic activity to the current period, which could in turn affect current crime ($Y$), thus violating the exclusion restriction. In this context, a different variable that displaces economic activity can be used as an NCI. For example, payday cycles ($NC$) are known to impact the timing of economic activity \citep{Hastings2010}. Payday timing is independent of weather. Therefore, an association between payday and crime would imply a violation of the exclusion restriction assumption. In this case, no conditioning on $Z$ is needed.\footnote{The NCI assumption is that payday timing only correlates with crime only through the timing of economic activity.}\endQuotable{}
By contrast, in contexts where a violation of outcome independence is suspected, the IV and the NCI are typically associated as well (as in Panel C of \ref{fig:dags_intro}). Therefore, the NCI test should condition on the IV.

As a result, unconditional independence tests between a negative control and the outcome are unique to IV settings. In non-IV settings, there is no exclusion restriction, and therefore, independence tests between a negative control and the outcome, carried out to detect unmeasured confounding, are always done conditionally.\footnote{The analog of NCI in non-IV settings is negative control exposure (NCE). NCE tests always condition on the exposure.}

\beginQuotable[AlwaysControl]
Nevertheless, researchers can choose to always control for the IV. Since both $NC$ and $Z$ are observed, the condition $NC \indep Z$ can be empirically tested. However, researchers might opt to skip this test and condition on the IV anyway. In a linear model, adding an additional control that is uncorrelated with $NC$ will not affect the coefficient estimate for $NC$ asymptotically. Furthermore, if $Z$ has a causal effect on $Y$,  including it in the regression can improve the precision of the estimation.\endQuotable{}

\subsection{Control Variables and Functional Forms}
 \label{SubSubSec:FuncForm}
%\textbf{Control Variables}

\beginQuotable[controls]
In many cases, the IV is believed to be valid only conditionally on certain control variables. For example, in papers that use judge assignment as an IV, the assignment of judges is quasi-random only within date and location \citep[e.g.,][]{kling2006incarceration}. Therefore, the independence assumption is satisfied only conditionally, and the IV design is valid only once controlling for date and location. 
\endQuotable{}

Formally, let $C$ be the vector of controls. Similar to the case without controls, outcome independence and exclusion restriction together imply $Z \indep Y(x)|\ C$. \ref{AppSec:proofs} presents the theory of negative controls when control variables are included.

When the IV is presumably valid only conditional on a vector of control variables $C$, an NCO test is a test for the null hypothesis 
\begin{equation}
\label{Eq:NCOnullWithControls}
H_0: NC \indep Z | C.
\end{equation}
Similarly, for NCIs, the null hypothesis is 
\begin{equation}
\label{Eq:NCEullWithControls}
H_0: NC \indep Y | C,Z.
\end{equation}
While accounting for controls in an IV analysis can be done in a variety of ways \cite[e.g.,][]{abadie2003semiparametric}, the large majority of applications use a two-stage least squares (2SLS) specification. This specification makes additional functional form assumptions. Most negative control tests used in practice adopt the same functional form as the 2SLS. 

In particular, NCO tests typically adopt the functional form for how the IV depends on the control variables. 
To avoid excessive notation, let $C$ also denote the set of controls in a 2SLS specification.\footnote{The vector $C$ may include, for example, a quadratic function of one of the original controls or interactions. For ease of notation, $C$ would always include the intercept.}
\beginQuotable[WeaklyCausal]
\cite{blandhol2022tsls} show that 2SLS requires the following linearity assumption to satisfy their definition of a \emph{weakly causal estimand}.\footnote{
A weakly causal estimand is a positively weighted average of subgroup-specific treatment effects.}
\endQuotable{}
\begin{ass}
[Rich covariates]\label{ass:RC}
The conditional expectation of the IV is linear in the control specification. Namely,
%\label{Eq:IVlinInControls}
  $\bbE[Z|C] =  \gamma'_CC ,$
for some vector $\gamma_C$.
\end{ass}
Combining the null hypothesis of NCO tests \eqref{Eq:NCOnullWithControls} and rich covariates, we expect that
\begin{equation}
\label{Eq:NCnull2SLS}
\bbE[Z|C,NC]= \gamma_C'C.
\end{equation}
This equation provides a more specific null hypothesis for conditional independence testing. This hypothesis can be tested by regressing the IV on the vector of controls and the NCO. The following corollary formalizes this argument. 

\begin{corollary}
\label{Coro:NCO_RC}
Assume that the random variable $NC$ is an NCO. Let
$$\gamma =(\gamma'_C,\gamma_{NC}) = \argmin_{b_c,b_{NC}}\bbE[Z-b'_CC - b_{NC}NC]^2$$
be the population-level OLS coefficient of regressing $Z$ on $C,NC$.
If $\gamma_{NC}\neq 0$, then either outcome independence, exclusion restriction, or rich covariates is violated.
\end{corollary}
%The proof is given in \ref{AppSubSec:corollaries}. 

In many cases, researchers run the reverse regression in which the NCO is the outcome variable. This practice is equivalent, as formalized in the following corollary. 

 \begin{corollary}
 \label{Coro:pseudo-outcome}
 Assume that $NC$ is an NCO. Let $$\beta =(\beta_Z,\beta'_C) = \argmin_{b_Z,b_C}\bbE[NC-b_ZZ - b'_CC]^2$$ be the population-level OLS coefficient of regressing $NC$ on $Z,C$. If $\beta_Z \neq 0$, then either outcome independence,  exclusion restriction, or rich covariates is violated. 
 \end{corollary}
%The proof is given in \ref{AppSubSec:corollaries}. 

NCI tests typically adopt the functional form of the relationship between the outcome and the IV and the control variables.
Specifically, NCI tests often use the same structure as the reduced form equation. Therefore, they implicitly make the following assumption.
\begin{ass}[Correctly Specified Reduced Form (CSRF)]
\label{ass:linear_RF}
The conditional expectation of the outcome is linear in the IV and the control variables. Namely,  $\bbE[Y|Z,C] = \theta_Z Z +  \theta'_CC$, for some $\theta_Z$ and vector $\theta_C$.
\end{ass}
Combining the null hypothesis \eqref{Eq:NCEullWithControls} with the CSRF assumption, we expect that 
\begin{equation}
\label{Eq:NCInullCSRF}
\bbE[Y|Z,C,NC]= \theta_Z Z +  \theta'_CC.
\end{equation}
This equation also provides a more specific null hypothesis, which can be tested with OLS. The following corollary shows that such an OLS jointly tests IV violations due to an alternative path and CSRF.

\begin{corollary}
\label{Coro:pseudo-IV}
Assume that the random variable $NC$ is an NCI. Let $$\theta = (\theta_Z,\theta'_C,\theta_{NC}) = \argmin_{b_Z,b_c,b_{NC}}\bbE[Y-b_ZZ-b'_CC - b_{NC}NC]^2$$
be the population-level OLS coefficient of regressing $Y$ on $Z,C,NC$.
If $\theta_{NC} \neq 0$ then either outcome independence,  exclusion restriction, or CSRF is violated. 
\end{corollary}
%The proof is given in \ref{AppSubSec:corollaries}. 
Corollaries \ref{Coro:NCO_RC}, \ref{Coro:pseudo-outcome}, and \ref{Coro:pseudo-IV} imply that, in the tests discussed, the null hypothesis can be rejected in IV designs that satisfy outcome independence and exclusion if functional form assumptions are violated.
For NCO tests, the null can be rejected because the rich covariates assumption is not satisfied. In such cases, researchers can still estimate a causal effect by modifying the functional form or using methods other than 2SLS \citep{blandhol2022tsls}.
For NCI tests, the null can be rejected because the CSRF assumption is violated. However, unlike rich covariates, CSRF is not a necessary assumption for 2SLS analysis, implying that negative control tests sensitive to this assumption could reject perfectly valid IV designs. 

For example, an NCI test can reject the null in designs where the IV is randomly assigned due to CSRF violation. Random assignment guarantees that outcome independence and rich covariates hold. Assuming the exclusion restriction holds, the design is valid. However, CSRF could still be violated if the IV has a nonlinear effect on the outcome or a heterogeneous effect across control vector values. In such cases, an NCI test may reject the null in \eqref{Eq:NCInullCSRF}, despite the design being valid.

\section{Implementation Guidance} \label{sec:practice}
This section offers guidelines for the implementation of negative control tests. We recommend that researchers follow four steps, summarized in \ref{fig:steps}. 
First, when possible, researchers should articulate specific threats to the validity of the IV design and characterize the alternative paths variables, as discussed in Section \ref{subsec:APV_guide}.
Second, researchers should survey available data to identify suitable negative controls---variables that can serve as proxies for the unobserved alternative path variables. 
These proxies should satisfy the NCO or NCI assumption (see Definitions \ref{dfn:NCO_var} and \ref{dfn:NCI_var}). Examples are discussed in Section~\ref{SubSec:chooneNC}. 

Third, researchers should choose a statistical test for independence between the NCO and the IV or between the NCI and the outcome, conditioning on the IV. \label{pt:controls}For IV designs that require conditioning on a set of controls, negative control tests should also condition on these controls. Section \ref{SubSec:choosingTest} discusses particular test specifications, their validity, and the assumptions they test, which in some cases also include functional form assumptions.
Fourth and finally, researchers should interpret the result and conduct further diagnostics if the test rejects the null, as discussed in Section \ref{SubSec:aftermath}.

To illustrate these recommendations, we apply them to IV designs used in prior work. We chose four widely cited papers published in the \textit{American Economic Review} with publicly posted replication data. We use \cite{autor2013china} and \cite{deming2014using} to discuss NCO tests and \cite{ashraf2013out} and \cite{nunn2014us} to discuss NCI tests.\footnote{\cite{ashraf2013out} and \cite{nunn2014us} are the two most cited AER papers published since 2013 that use an NCI test. Similarly, \cite{autor2013china} is the most cited AER paper published after 2013 that uses an NCO test. \cite{deming2014using} was selected to demonstrate how our proposed follow-up analysis can be used to diagnose and correct problems with the IV design in Section \ref{SubSec:aftermath}.} \ref{tab_applications_intro} summarizes the IV designs in these papers and the negative controls they used in their falsification tests. \ref{appendix_practice} provides additional details on our analyses.

\subsection{Potential Threats}\label{subsec:APV_guide}

Two guiding questions can help researchers characterize potential violations of IV validity. This characterization can assist in selecting appropriate negative control variables and determining which hypothesis should be tested. The first question is whether the primary concern is a violation of outcome independence (Assumption \ref{Ass:Indep}) or the exclusion restriction (Assumption \ref{Ass:ExclRestriction}). Both types of violations introduce an alternative path between the IV and the outcome. 

Outcome independence is violated when this path is through some factor that affects both the IV and the outcome. As previously discussed, \citet{martin2017bias} examine whether Fox News viewership influences Republican vote shares using cable channel positions as an IV. The concern is that cable companies may place Fox News in lower channel numbers in conservative locations, where voters lean republican regardless. Violations of outcome independence are illustrated in panels A and C of \ref{fig:dags_intro}. 

The exclusion restriction is violated when the IV affects the outcome through channels other than its effect through the treatment. For example, as previously discussed, in \citet{angrist1996children}, the sex composition of the first two children (the IV) may affect female labor force participation (the outcome) not only through its effect on family size (the treatment) but also through its effects on household expenditures due to hand-me-downs. Exclusion restriction violations can occur even with randomly assigned IVs, as in a randomized controlled trial. Violations of the exclusion restriction are illustrated in panels B and D of \ref{fig:dags_intro}.

The second question is whether the threat (the alternative path) operates through an APO or an API variable. That is, does the alternative path variable have a known association with the outcome, and the concern is that it may also be associated with the IV? Or does it have a known association with the IV, and the concern is that it may also be related to the outcome? In the first case, the alternative path operates through an APO variable; in the second, it operates through an API variable. In the context of \citet{martin2017bias}, unobserved conservativeness is an APO variable---it certainly influences voting for Republican candidates (the outcome), yet it is unclear whether it is also associated with Fox News channel placement (the IV). By contrast, in \citet{nunn2014us}, weather conditions are an API variable---they surely affect wheat production (the IV), and the concern is that they may also directly affect the conflicts in aid recipient countries (the outcome). Note that both outcome independence and exclusion restriction can be violated through either APO or API variables. In \ref{fig:dags_intro}, panels A and B demonstrate this for APO variables and panels C and D for API variables. 

These two questions can assist researchers in finding relevant negative control variables and choosing the right negative control tests. For APO variables, researchers should search for NCOs and test their association with the IV. For API variables, researchers should search for NCIs and test their association with the outcome, conditionally on the IV. The type of violation (outcome independence or exclusion restriction) can be useful for thinking of relevant negative control types.

\subsection{Choosing Negative Controls} \label{SubSec:chooneNC}

In this section, we discuss different types of negative controls, both commonly used in practice and novel ones suggested by the theoretical framework.

\subsubsection{Common Types of Negative Control Outcomes}

\beginQuotable[Predetermined]
\textbf{Predetermined Variables.} 
%A paragraph that everyone can understand 
Variables fixed before the IV is determined are frequently used in NCO tests. Common examples include lagged outcome variables and demographic characteristics such as gender, race, and age.
%A paragraph that you need to read the paper to understand
Predetermined variables are useful for testing outcome independence. 
In some cases, researchers may choose to use predetermined variables as NCOs even without clear knowledge of which exact APO variable they proxy for. If the IV is associated with a predetermined variable, it could imply that it is affected by something that affects the outcome as well. 
\endQuotable{}

However, not every predetermined variable is a valid NCO. First, NCOs need to satisfy U-comparability (Definition~\ref{dfn:NCO_var})---predetermined variables that are completely unrelated to the outcome are uninformative and should not be used. Second, not all predetermined variables satisfy the NCO assumption. In particular, certain predetermined variables may influence the IV, even if the underlying IV design is valid. For example, when the IV is the child's quarter of birth \citep[as in][]{angrist1991does}, the parents' quarter of marriage is not a valid NCO as it likely influences the child's quarter of birth, even if the design is valid.

\beginQuotable[DemingPredetermined]
When IVs are assumed to be quasi-randomly assigned (e.g., lotteries), they cannot be affected by predetermined variables.
Researchers could therefore use predetermined variables as NCOs to evaluate the claim that the IV is quasi-random.\footnote{In this case, predetermined variables are NCOs based on the more general Definition \ref{dfn:NCO_var_controls}. This definition allows for the NCO to be directly associated with the IV if the IV is not quasi-random as claimed.}
If the IV is associated with a predetermined variable, it is unlikely to be quasi-random, and hence, outcome independence might not hold.
For example, we found multiple predetermined variables in the replication data from \citet{deming2014using}, which uses an IV constructed based on school lotteries. We use these predetermined variables as NCOs to evaluate outcome independence. We have found an association of the IV with the predetermined variables. In particular, we found that the construction of the IV involved non-random components that require additional controls; see Section~\ref{SubSec:aftermath}.
\endQuotable{}

\beginQuotable[AutorPredetermined]
Predetermined variables are also useful NCOs when the IV is not quasi-random.  
For example, \citet{autor2013china} use a shift-share IV for commute-zone exposure to Chinese imports to evaluate their impact on employment. To evaluate this IV, they use predetermined local labor market manufacturing employment as NCOs. The concern is that since industry exposure to Chinese imports is non-random, it might be associated with other labor market conditions, which in turn could be associated with the outcome (e.g., Chinese imports are more pronounced in regions with industries that were declining in Western countries regardless). Such an association would violate outcome independence. 
We found many additional predetermined variables in the original paper's replication data that could proxy for latent local labor market conditions (e.g., past unemployment). These variables can also serve as NCOs. The NCO assumption requires that any association of the IV with the predetermined variables used as NCOs is driven by an APO variable, i.e., by something that also affects the outcome. This would be violated if, for example, Chinese import penetrated industries due to economic factors that were only relevant in the past and are no longer relevant in the studied period. 
\endQuotable{}

\textbf{IV Leads and Lags.} 
Certain IVs are predicated on serendipitous or chance occurrences (``strokes of luck''). Because such unexpected shocks should not be autocorrelated, leads and lags of the variable used as the IV can serve as NCOs.  For example, \cite{jager2022substitutable} use a worker's premature death as an IV for employee turnover, under the assumption that such deaths occur randomly across firms. A potential concern is that deaths are non-random and reflect riskier conditions in the firm (the APO variable) that directly impact wages (the outcome). This would violate outcome independence. To rule this out, \citeauthor{jager2022substitutable} use subsequent premature deaths in the same firm as an NCO that proxies for potentially unobservable riskier conditions. The NCO assumption here stipulates that given the risk conditions, premature deaths should not be autocorrelated. A recurring pattern of premature deaths would cast doubt on the assumption that such deaths occur randomly across firms.

\textbf{Alternative Outcomes.}
Alternative or unrelated outcomes can also serve as NCOs for two different types of APO variables. First, APO variables can potentially affect the IV, forming an alternative path that violates outcome independence (as in Panel~A of \ref{fig:dags_intro}). Alternative outcomes that are affected by the same APO variable can then serve as NCOs. For example, \cite{chetty2014measuring} leverage teachers' moves between schools to evaluate middle-school teacher value-added measures. The concern is that high-quality teachers may tend to move to schools that experience simultaneous improvements in student quality. Here, the APO variable is the unobserved changes in school quality. To evaluate this threat, \citeauthor{chetty2014measuring} use as NCOs test scores from subjects not taught by the teacher in question. If the NCO test finds that teacher quality is associated with better outcomes in subjects they do not teach, it would cast doubt on the design's validity. \citeauthor{chetty2014measuring} focus on middle-school teachers, as opposed to elementary school teachers who teach multiple topics. This is because the NCO assumption requires that the IV will not affect the alternative outcome directly. The IV should also not affect alternative outcomes indirectly via the treatment or outcome. 

The second type of APO variables potentially violates the exclusion restriction. The concern is that the IV affects an additional factor (the APO variable), which in turn affects the outcome (as in Panel B of \ref{fig:dags_intro}).
In the previously discussed example of \citet{angrist1996children}, the concern is that same-sex sibship IV may affect female labor supply due to hand-me-downs, thus forming an alternative path. To evaluate this concern, \citet{rosenzweig2000natural} check the correlation of same-sex sibship and an alternative outcome---clothing expenditure.\footnote{\beginQuotable[Rosenzweig]
In this example, the NCO is the APO variable itself, so the NCO assumption is trivially satisfied\endQuotable{} 
(see Section~\ref{SubSubSec:NCassumptions})}

\subsubsection{Common Types of Negative Control Instruments}
\textbf{Variables Similar to the IV That Do Not Affect the Treatment.} 
Researchers often choose NCIs that are similar to the IV but are presumed not to influence the treatment variable. These NCIs typically test outcome independence. They usually share many similarities with the IV and are thus likely to be correlated with the API variable. For example, as previously discussed (and illustrated in Panel C of \ref{fig:dags_intro}), \cite{nunn2014us} use US wheat production as their IV for US aid and consider the US production of other crops unrelated to US aid (e.g., oranges) as NCIs. These variables are similar, as they are both affected by the same API variables such as weather conditions. Similarly, \citet{ashraf2013out} replace their original IV, distance from Addis-Ababa, with distance from London, Tokyo, and Mexico City. 

\beginQuotable[SimulatedIV]
In some cases, researchers generate variables similar to the IV on their own. They construct a variable in a similar way to how the IV was constructed but remove the impact on the treatment. For example, \citet{de2020consumption} study the effect of peers' consumption on own consumption. As an IV, they use economic shocks to firms of distant peers. This IV will be correlated with shocks to large firms (as statistically, they are more likely to affect all workers, including distant peers), which could potentially affect the outcome in other ways. To test this, they use an NCI which they call a ``placebo'' IV---they calculate the same IV when replacing the real allocation of workers to employers with a random allocation, keeping firm sizes constant. 
\endQuotable{}

\textbf{IV Leads.} 
Future instances of the IV (IV leads) can often serve as effective NCIs for testing outcome independence. For example,  \citet{moretti2021effect} studies the effect of the size of high-tech clusters on productivity. As an IV, he uses predicted cluster size based on the expansion of local firms outside the cluster. \citeauthor{moretti2021effect} then ascertains that future predicted cluster size is also not correlated with productivity. This relies on the fact that IV leads, which are based on events that occur after the outcome, cannot influence it. IV leads share similarities with the IV and are therefore likely to be associated with the API variable. To satisfy the NCI assumption, the outcome must not influence future realizations of the IV. In the example of \citet{moretti2021effect}, regional productivity cannot affect the expansion of local firms in other locations.

When practitioners observe IV leads, they need to consider whether they expect the IVs to be autocorrelated. When the IVs are expected to be autocorrelated \citep[as in][]{moretti2021effect} an NCI strategy can be used. When the IVs are presumably uncorrelated, and NCO strategy can be applied, as discussed in the previous section.

\subsubsection{Underutilized Types of Negative Control Instruments}
\beginQuotable[CauseofIV]
\textbf{Causes of the IV.} In some cases, researchers may suspect violations of outcome independence through API variables that affect the IV and potentially also affect the outcome. In these cases, researchers can use variables that causally affect the IV as NCIs. This approach does not require full knowledge of the API variable.  In the example from \cite{angrist1991does}, the parents' quarter of marriage influences the child's quarter of birth (the IV) and qualifies as a valid NCI. In this case, API variables are any factors that influence a child's quarter of birth and are suspected to affect wages (the outcome).

Panel~B of \ref{fig:scenarios} illustrates how such NCIs work. When both the API variable and the NCI influence the IV, they are associated conditional on the IV \citep[in this DAG, the IV act as a \emph{collider};][]{pearl2009causality}. If,  conditional on the IV, the NCI is also associated with the outcome, it implies that a path exists between the NCI and the outcome through the IV and the API variable. This, in turn, implies that the API variable affects the outcome, violating outcome independence. 
\endQuotable{}

\beginQuotable[NCIAssumptionCauseIV]
To satisfy the NCI assumption, these NCIs should have no association with the outcome other than through the IV.\footnote{If such an association does exist, these variables must be used as controls.} In particular, the NCI cannot directly affect either the treatment or the outcome. In the example of \citet{angrist1991does}, using parents' quarter of marriage as an NCI requires assuming that marriage timing does not directly influence child schooling or wages. 
\endQuotable{}

\textbf{IV Side Effect Proxies.} An IV that not only affects the treatment but also produces a side effect may violate the exclusion restriction. This occurs if the side effect also affects the outcome. In such cases, proxies for the side effect (the API variable) may serve as NCIs. 

\beginQuotable[NCIExcRestExample]
Panel D of \ref{fig:dags_intro} illustrates a scenario where the IV influences the NCI through the API variable (therefore, the NCI is itself a side effect). 
For example, in the previously discussed context of \citet{jacob2007crime}, lagged weather ($Z$) serves as an IV for lagged crime ($X$) to study its impact on current crime ($Y$). Lagged weather also creates intertemporal displacement of economic activity (the API variable $U_4$). The concern is that intertemporal displacement of economic activity affects subsequent crime, thus violating the exclusion restriction. To evaluate whether this alternative path exists, Jacob et al. use traffic patterns, a proxy for economic activity, as an NCI ($NC_4$). Testing whether traffic patterns are correlated with crime, conditional on lagged weather, constitutes an NCI test for this alternative path. If a correlation exists, it suggests that the exclusion restriction is violated, as weather influences crime not only through past crime but also through displaced economic activity.\endQuotable{}
Alternatively, Panel~A of \ref{fig:scenarios} presents another type of side-effect proxy that influences the API variable rather than being influenced by it. For  \citeauthor{jacob2007crime}, that could be other factors that displace economic activity (e.g., payday schedule; see the discussion of this example in Section~\ref{SubSubSec:NCOtest}).

\subsubsection{Power Considerations When Choosing Negative controls}  
\label{subsubsec:NC_stat_power}

Negative control variables must satisfy U-comparability. This implies that these variables are indeed associated with the alternative path variable. Some negative controls might satisfy this condition but have only a weak association with the alternative path variable. In this case, if the IV design is not valid, the association between the NCO and the IV or the NCI and the outcome would be difficult to detect without having a large dataset. Therefore, power considerations suggest excluding negative control variables that have only a weak association with the alternative path variable, as they can lower test power. This mirrors the effect of irrelevant control variables in OLS. This issue is especially acute in NCI variables that are intentionally similar to the original IV. Such variables are often strongly correlated with the IV but only weakly correlated with the API variable conditional on the IV.

\subsection{Choosing a Statistical Test}
\label{SubSec:choosingTest}

As discussed in Section~\ref{SubSubSec:NCOtest}, negative control tests assess whether a negative control variable is (conditionally) independent of the IV or outcome. %, conditional on the controls from the original IV specification. For NCI tests, conditioning on the IV itself is typically necessary. 
%Conditional independence, especially with continuous controls, has no established one-size-fits-all solution \citep{Shah2020hardness}. 
The choice of statistical test for conditional independence should take into account three primary considerations: the estimation method (e.g., 2SLS); the anticipated functional form of the relationship between the negative control and either the IV or the outcome (and the controls); and statistical power.
This section discusses both commonly used and underused conditional independence tests and presents examples using replication data from existing work.

\subsubsection{Commonly Used Tests}\label{SubSec:parametric}

The most commonly used NCO falsification tests are based on the original IV reduced-form equation, $Y = \alpha_Z Z +  \alpha'_CC + \epsilon$, but replace the outcome with an NCO (e.g., past outcomes). That is, the test estimates the model
\begin{equation}
\label{Eq:pseudo_outcome}
NC = \beta_Z Z +  \beta'_CC + \epsilon_{NC},
\end{equation} 
and evaluates the null hypothesis $H_0:\beta_Z=0$. This test evaluates outcome independence or exclusion restriction, as well as the rich covariates assumption (Corollary \ref{Coro:pseudo-outcome}). Because in 2SLS the rich covariates assumption is necessary for causal interpretation, this test provides useful information regarding both the validity of the IV design and of the 2SLS specification.
Since this test uses the same inferential framework as the original study, it can also expose errors in the inference method \citep{eggers2021placebo}. 

\beginQuotable[MulHyp]
When multiple negative controls are available, Model \eqref{Eq:pseudo_outcome} can be estimated separately for each negative control. However, correcting for multiple hypotheses is necessary, which reduces statistical power. An alternative approach is to jointly incorporate multiple negative controls using the model
\begin{equation}
\label{Eq:LinReg}
Z= \gamma'_{NC}NC +\gamma'C +  \epsilon_Z,
\end{equation}
where $NC$ now represents a vector of NCOs.\endQuotable{}\footnote{In most applications, if a vector of negative controls is associated with the IV, at least one of its components will be as well. \ref{AppSubSec:NCVecNotNC} provides a theoretical counterexample, but such cases are unlikely in practice as small parameter changes would reverse the result.}
\beginQuotable[Ftest]An F-test can be used to evaluate the null hypothesis 
$H_0: \gamma'_{NC}=0'$ under standard assumptions.\footnote{With robust standard errors, the common implementation calculates the Wald statistic, divides it by the degrees of freedom, and calculates a p-value from an F-distribution.}\endQuotable{}

For NCI tests, a similar approach applies, but the outcome is regressed on the NCI, controlling for the IV. Specifically, the NCI (denoted again NC) is added to the reduced-form equation
\begin{equation}\label{Eq:pseudo-IV-conditional}
    Y =  \theta_ZZ + \theta'_CC + \theta_{NC}NC +   \epsilon_{Y},
\end{equation}
and the null hypothesis is $H_0:\theta_{NC} =0$.
A common mistake in practice is omitting the IV from this regression (see Section \ref{sec:survey}). Doing so can lead researchers to reject the null, even when the IV design is valid. The IV can be omitted in the (rare) event that the NCI is independent of the IV.

This NCI test can still reject the null for valid IV designs (even when conditioning on the IV). Corollary \ref{Coro:pseudo-IV} shows that even if the IV design is valid, the null could still be rejected due to a violation of CSRF (Assumption~\ref{ass:linear_RF}). As discussed in Section \ref{SubSubSec:FuncForm}, the CSRF assumption is not necessary for causal identification and can be violated even under random assignment.\footnote{An exception is binary IV without controls, in which case CSRF is always satisfied.} 
Therefore, tests based on Model \eqref{Eq:pseudo-IV-conditional} may reject the null hypothesis, even though the IV design can identify a causal estimand with 2SLS. To relax this sensitivity to linearity assumptions, researchers may opt for semi-parametric or non-parametric tests, which are discussed next.

\subsubsection{Additional Tests}
\label{SubSec:semiparametric}

The tests discussed above inherit the basic functional form assumption as in the 2SLS specification. The NCO tests replace the outcome with the NCO in the reduced-form equation or posit the reverse linear model for the IV. The NCI tests include the NCI additively in the reduced form. Moreover, these tests only test the mean independence of the IV with the NCO or the outcome with the NCI. In some cases, researchers should use more general tests.

\beginQuotable[nonLinear]
Researchers should consider replacing the functional form assumptions in two cases. First, tests with more flexible functional forms can help relax undesirable functional form assumptions. In particular, when using an NCI, researchers should test for non-linear associations whenever possible to avoid testing the CSRF assumption, which is unnecessary. By contrast, the previously discussed linear NCO tests examine the rich covariates assumption, which is necessary for 2SLS.
Therefore, tests based on Models  \eqref{Eq:pseudo_outcome} or \eqref{Eq:LinReg} are often preferable in such contexts.

Second, more flexible tests are useful when researchers suspect a non-linear association between the NCO and the IV or between the NCI and the outcome. Such non-linear associations also indicate a violation of the IV assumptions and should therefore be tested whenever possible. For example, in studies using crops as an IV \citep[as in][]{nunn2014us}, one might be concerned that extreme weather conditions, such as unusually high or low temperatures, could affect both crop yield and conflict incidence (the outcome). In this case, Model \eqref{Eq:pseudo-IV-conditional} might fail to detect a non-monotonic relationship between an average temperature NCI and the outcome.

To estimate more complex functional forms, researchers can include higher-order polynomial terms or interactions in regression models. Alternatively, they can use semi- or non-parametric tests. The next section discusses a few examples, which we implement in our application examples.  
%These additional tests examine more complex functional forms for mean independence.

When using estimation methods other than 2SLS, researchers should consider more general conditional independence tests that assess relationships beyond mean independence. For example, IV quantile regression \citep[e.g.,][]{chernozhukov2008instrumental} relies on the broader notion of conditional independence. In such cases, researchers can use quantile regression of the IV on the NCO or the outcome on the NCI, with appropriate controls. A variety of other conditional independence tests can also be considered \citep[see][for recent surveys]{heinze2018invariant, li2020nonparametric}.
The choice of test depends on the specific context, as there is no uniformly optimal conditional independence test.\footnote{\citet{Shah2020hardness} show that for continuous distributions, there is no conditional independence test that is uniformly valid and is simultaneously powerful against any conditional dependence types.} 
Naturally, more complex tests require larger datasets or low-dimensional covariates for reliable implementation. Therefore, these tests may be less informative in small samples or in settings with many covariates. %In such cases, it might be preferable to incorporate suspected non-linear relationships directly into the model (e.g., by including polynomial terms).
\endQuotable{}  

\subsubsection{Examples of Applications} 
\label{subsubsec:examples_of_applications}

\ref{tab_applications_nco} presents the NCO tests results using data from \cite{autor2013china} and \cite{deming2014using}. Column~(2) shows that a test based on a single NCO fails to reject the null hypothesis in both cases.\footnote{For \cite{autor2013china}, Column~(1) replicates the original NCO test from their paper, without control variables. They find the IV is significantly associated with the lagged outcome, albeit with the opposite sign from the main analysis. This association becomes insignificant when all controls are included.} However, Column~(3) demonstrates that using multiple NCOs and applying a Bonferroni correction leads to rejection of the null, as does a joint F-test (Column~4). These results underscore how theory-guided inclusion of additional NCOs enhances test power. 

To examine non-linear associations, we implement the common semi-parametric approach of Generalized Additive Models \citep[GAMs; see][]{hastie1990generalized,wood2006generalized}. We express the IV as an additive combination of smooth functions of the controls ($C$) and NCOs ($NC$):  
\begin{equation}
\label{Eq:GAM}
Z =  \sum_j f_j(C_j) + \sum_k g_k(NC_k) + \epsilon_Z,
\end{equation}  
where $f_j$ and $g_k$ are smooth functions estimated via splines. We test whether $g_k = 0$ for all $k$ to assess whether the NCOs are conditionally independent of the IV.\footnote{A similar GAM test can be implemented for NCI tests, by adding smooth functions to Model~\eqref{Eq:pseudo-IV-conditional}.}

To test rich covariates (Assumption~\ref{ass:RC}) within this framework, researchers can restrict $f_j$ to be linear (i.e., $f_j(C_j) = \gamma_j C_j$) while allowing $g_k$ to remain nonlinear. 
Such a model is useful when using 2SLS, which requires rich covariates, while suspecting a strong nonlinear association between the NCO and the IV. 
Columns~(5) and~(6) in \ref{tab_applications_nco} implement a GAM test with and without assuming linearity in the controls. As expected, the GAM test underperforms in smaller samples.  

\ref{tab_applications_nci} presents the NCI test results for \cite{nunn2014us} and \cite{ashraf2013out}. Columns~(1) and~(2) show the results of regressing the outcome on the NCI, both with and without conditioning on the IV. In both studies, conditioning on the IV is necessary. 
For \cite{nunn2014us}, failure to condition leads to rejection of the null, suggesting a false rejection due to the path between the NCI and the outcome through the IV. For \cite{ashraf2013out}, the null is not rejected in either case, likely due to limited statistical power. 
Columns~(3) and~(4) of \ref{tab_applications_nci} implement multiple separate NCI tests with Bonferroni corrections and joint F-tests using all NCIs. In both approaches, the null is not rejected.\footnote{Due to the small sample sizes relative to the number of control variables, proper estimation of GAM models is infeasible for both studies. For \cite{nunn2014us}, we estimate a GAM model assuming linear controls; see \ref{SubSec:NunnQianDetails}.}

\subsection{Interpreting the Test Results} \label{SubSec:aftermath}

\textbf{Rejection of the Null.}\label{sec:failures}
Rejecting the null hypothesis in a parametric linear negative control test, such as an F-test, may indicate a violation of the IV assumptions or failure of the linearity assumption (rich covariates for NCO, CSRF for NCI). Researchers can further investigate by directly testing the linearity assumption \citep[e.g., using a Regression Equation Specification Error Test (RESET) with control variables only;][]{ramsey1969tests}. With sufficient sample size, semi- or non-parametric tests can also test outcome independence or the exclusion restriction without relying on strict functional form assumptions.

When using multiple negative controls, identifying which ones drive the rejection can provide important insights. A diagnostic scatter plot of the correlation of each negative control with the IV against its correlation with the outcome can highlight potential alternative pathways.
%Santos asked to expand on the control variables for the examples and this is a good place to do it so I took it out of the footnote. 
\ref{corplot_lott_VA.pdf} displays such diagnostics for \cite{deming2014using}. \beginQuotable[demingControl]\cite{deming2014using} uses predicted school value-added, based on school lotteries, to evaluate school value-added measures. Specifically, the IV is the value added of the student's preferred school if they won the lottery and of their neighborhood school if they did not (see \ref{tab_applications_intro} for details). This IV satisfies outcome independence only when controlling for the relevant school value-added measures and the probability of winning the lotteries. The scatter plot reveals that the NCO with the strongest correlation with the IV is the value added of the neighborhood school. This occurs because the original 2SLS analysis does not control for the neighborhood school value added. This diagnostic thus identifies a fixable problem in the IV construction. This problem is resolved when using the original lottery results as an IV, as outcome independence is satisfied conditional only on the winning probabilities (i.e., the schools the student applied to).\endQuotable{}\footnote{While estimates using the original lottery as the IV are noisier, we cannot reject the main conclusions.}

One caveat of this diagnostic exercise is that correlations with negative controls may not directly reflect the strength of alternative paths. Negative controls are proxies, and so the strength of their correlations with the IV or outcome depends on the strength of their correlation with the alternative path variables. Thus, weak correlations between a negative control and the IV or outcome might still mask strong alternative paths.

\textbf{Non-Rejection of the Null.}
\label{subsubsec:success}
As discussed in Section \ref{Sec:theory}, failure to reject the null does not imply that the IV is valid. Two concerns remain. First, the IV may still be invalid due to alternative path variables not captured by the NCO or NCI used in the test. For example, a quasi-random allocation to teachers that is found to be uncorrelated with students' neighborhoods could still be correlated with students' abilities within neighborhoods. 
Second, an invalid IV design may pass the test due to limited statistical power.

\section{Conclusion}
\label{sec:conclusion}
This paper provides a thorough examination of the assumptions underlying negative control tests for IV designs. Our analysis clarifies existing practices and emphasizes several issues of direct practical relevance.
First, most current implementations of NCI tests fail to condition on the original IV, which could lead to the unwarranted rejection of valid IV designs. 
Second, common negative control tests assess not only the outcome independence and exclusion restriction assumptions but also assess specific functional form assumptions. Because these assumptions are replaceable and sometimes unnecessary, researchers should distinguish between the essential IV identification conditions and the ancillary functional form assumptions when interpreting and considering additional tests. Third, our analysis clarifies what variables can serve as negative controls. These include variables that are rarely used in practice, such as variables that causally affect the IV. Moreover, in some cases, negative control variables are readily available in researchers' datasets and should be used to construct more powerful negative control tests. We hope this paper will foster a more systematic and efficient use of negative control falsification tests in empirical IV designs.

\beginQuotable[Discussion]
While this paper focused on the role of negative controls in testing the IV assumptions, negative controls can also be used for other purposes. As discussed, when the IV design is valid, NCOs can be used to test functional form assumptions. NCOs can also improve the estimation precision, for example, when included as control variables. Since NCOs are correlated with the outcome but not with the IV, including them as controls can improve precision. By contrast, NCIs cannot be used similarly. They do not test a necessary functional form assumption, as we discussed in Section \ref{SubSubSec:FuncForm}. Moreover, including NCIs in the reduced form will only decrease precision because they are correlated with the IV and not with the outcome.\footnote{In case of heteroskedasticity, at face value, NCIs may be used to improve precision in valid designs by including additional moment conditions for their orthogonality with the error as in \citet{cragg1983more}. However, this can also be achieved by including other functions of the IV.}
\endQuotable{}

%\clearpage
{%https://www.overleaf.com/project/6063438c42f1341b65bb15db
\setstretch{1.15}  % avoiding orphan reference on the last page

\bibliographystyle{aea}
\bibliography{nc}
}

%====================================== FIGURES =======================================%

%%% Figure - Illustration of NCO
\begin{figure}[p]
\caption{Negative Control Falsification Tests: Graphical Illustrations \label{fig:dags_intro}
}
\centering
\vspace{2em}
Negative Control Outcome Tests\\
\vspace{1em}
% Panel A
\begin{subfigure}{0.45\textwidth}
\centering
\begin{tikzpicture}
\node (1) {$Z$};
\node [above right = 0.75cm and 0.25cm of 1](2) {$W$};
\node [right =   1cm of 1](3) {$X$};
\node [right =   1cm of 3](4) {$Y$};
\node [left = 1cm of 1](5) {\color{red} $U_1$};
\node [below  = 0.5cm of 1](6) {\color{blue}$NC_1$};
\draw[Arrow] (1) -- (3);
\draw[Arrow] (2) -- (3);
\draw[Arrow] (2) -- (4);
\draw[Arrow] (3) -- (4);
\draw[DashedRedArrow] (5) to  (1);
\draw[Arrow] (5) to [out=-25, in=-135] (4.west);
\draw[Arrow] (5) -- (6);
%\draw[Arrow] (5) to [out=-25, in=180] (6);
\end{tikzpicture}\hspace{2em}
\caption{$NC_1\cancel{\indep}Z$ implies that the dashed arrow exists, thus violating the outcome independence.}
\label{fig:nco_dags_independence}
\end{subfigure}
\hfill
% Panel B
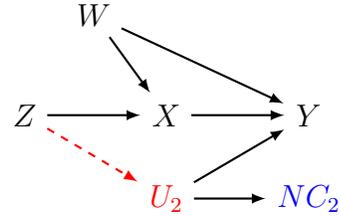
\begin{subfigure}{0.45\textwidth}
\centering
\begin{tikzpicture}
\node (1) {$Z$};
\node [above right = 0.75cm and 0.25cm of 1](2) {$W$};
\node [right =   1.25cm of 1](3) {$X$};
\node [right =   1.25cm of 3](4) {$Y$};
\node [below =   0.5cm of 4](6) {\color{blue}$NC_2$};
\node [below =   0.5cm of 3](5) {\color{red}$U_2$};
\draw[Arrow] (1) -- (3);
\draw[Arrow] (2) -- (3);
\draw[Arrow] (2) -- (4);
\draw[Arrow] (3) -- (4);
\draw[DashedRedArrow] (1) -- (5);
\draw[Arrow] (5) -- (4);
\draw[Arrow] (5) -- (6);
\end{tikzpicture}
\caption{$NC_2\cancel{\indep}Z$ implies that the dashed arrow exists, thus violating the exclusion restriction.}
\label{fig:nco_dags_exclusion}
\end{subfigure}

\vspace{3em}

Negative Control Instrument Tests\\
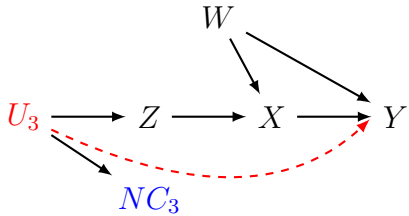
\vspace{1em}
% Panel C
\begin{subfigure}{0.45\textwidth}
\centering
\begin{tikzpicture}
\node (1) {$Z$};
\node [above right = 0.75cm and 0.25cm of 1](2) {$W$};
\node [right =   1cm of 1](3) {$X$};
\node [right =   1cm of 3](4) {$Y$};
\node [left = 1cm of 1](5) {\color{red} $U_3$};
\node [below  = 0.5cm of 1](6) {\color{blue}$NC_3$};
\draw[Arrow] (1) -- (3);
\draw[Arrow] (2) -- (3);
\draw[Arrow] (2) -- (4);
\draw[Arrow] (3) -- (4);
\draw[Arrow] (5) to  (1);
\draw[DashedRedArrow] (5) to [out=-25, in=-135] (4.west);
\draw[Arrow] (5) -- (6);
\end{tikzpicture}
\caption{$NC_3\cancel{\indep}Y|Z$ implies that the dashed arrow exists, thus violating the outcome independence.}
\label{fig:nci_dags_independence}
\end{subfigure}
\hfill
% Panel D
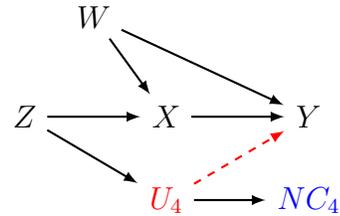
\begin{subfigure}{0.45\textwidth}
\centering
\begin{tikzpicture}
\node (1) {$Z$};
\node [above right = 0.75cm and 0.25cm of 1](2) {$W$};
\node [right =   1.25cm of 1](3) {$X$};
\node [right =   1.25cm of 3](4) {$Y$};
\node [below =   0.5cm of 4](6) {\color{blue}$NC_4$};
\node [below =   0.5cm of 3](5) {\color{red}$U_4$};
\draw[Arrow] (1) -- (3);
\draw[Arrow] (2) -- (3);
\draw[Arrow] (2) -- (4);
\draw[Arrow] (3) -- (4);
\draw[Arrow] (1) -- (5);
\draw[DashedRedArrow] (5) -- (4);
\draw[Arrow] (5) -- (6);
\end{tikzpicture}
\caption{$NC_4\cancel{\indep}Y|Z$ implies that the dashed arrow exists, thus violating the exclusion restriction.}
\label{fig:nci_dags_exclusion}
\end{subfigure}

\bigskip

\exhibitnotes{The figure illustrates how negative control tests assess the validity of IV designs. In all the panels, $X$ represents the endogenous treatment variable, $Y$ the outcome, $Z$ the IV, and $W$ a potential confounder that motivates the use of IV. The variables $U_i$ are unobserved variables that threaten identification when the dashed arrows exist. The top panels (A and B) depict negative control outcome tests. IV validity is threatened by the concern that $U_1$ or $U_2$ (\emph{alternative path outcome} (APO) variables) are related to the IV, thus violating the outcome independence (Panel A) or exclusion restriction (Panel B) assumptions. An observed negative control outcome ($NC_1$ or $NC_2$) related to each APO variable can be used to evaluate the presence of the problematic association by testing whether $NC_i \indep Z $. The bottom panels (C and D) depict negative control instrument tests, addressing concerns that $U_3$ or $U_4$ (\emph{alternative path instrument} (API) variables) might be related to the outcome, thus violating the outcome independence (Panel C) or exclusion restriction (Panel D) assumptions. An observed negative control instrument ($NC_3$ or $NC_4$) related to each API variable can examine these concerns by testing whether $NC_i\indep Y|Z$.}

\end{figure}

%%% Fig - scenarios NCI

\begin{figure}[p]
\caption{Illustration of Scenarios Related to Negative Control Instrument Tests\label{fig:scenarios}}

\begin{subfigure}{0.45\textwidth}
    \centering
\begin{tikzpicture}
\node (1) {$Z$};
\node [above right = 0.75cm and 0.25cm of 1](2) {$W$};
\node [right =   1.25cm of 1](3) {$X$};
\node [right =   1.25cm of 3](4) {$Y$};
\node [below =   0.5cm of 1](6) {\color{blue}$NC$};
\node [below =   0.5cm of 3](5) {\color{red}$U$};
\draw[Arrow] (1) -- (3);
\draw[Arrow] (2) -- (3);
\draw[Arrow] (2) -- (4);
\draw[Arrow] (3) -- (4);
\draw[Arrow] (1) -- (5);
\draw[DashedRedArrow] (5) -- (4);
\draw[Arrow] (6) -- (5);
\end{tikzpicture}
\caption{An unconditional NCI test}
\end{subfigure}
\hfill    
\begin{subfigure}{0.45\textwidth}
%\bigskip
\begin{tikzpicture}
\node[rectangle,draw] (1) {$Z$}; 
% Draw a circle around node 1 (Z)
%\draw[thick] (1) circle (0.25);
\node [above right = 0.25cm and 0.25cm of 1](2) {$W$}; 
\node [right =   1cm of 1](3) {$X$}; 
\node [right =   1cm of 3](4) {$Y$}; 
\node[left = 1cm of 1](5) {\color{red} $U$}; 
\node [above  = 0.25cm of 5](6) {\color{blue}$NC$}; 
\draw[Arrow] (1) -- (3); 
\draw[Arrow] (2) -- (3); 
\draw[Arrow] (2) -- (4); 
\draw[Arrow] (3) -- (4); 
%\draw<3->[RedArrow] (3) -- (6); 
\draw[Arrow] (5) to  (1); 
\draw[DashedRedArrow] (5) to [out=-25, in=-135] (4.west);
\draw[Arrow] (6) -- (1);
\end{tikzpicture}
\caption{NCI causally affecting the IV}
\end{subfigure}
\bigskip

\exhibitnotes{Each panel represents a different scenario related to negative control instrument (NCI) tests. 
In both scenarios, $X$ is the endogenous treatment variable, $Y$ is the outcome, $Z$ is the IV, and $W$ is a potential confounder that motivates the use of the IV. The validity of the IV design is challenged by the potential alternative paths. Panel~A demonstrates an NCI scenario where conditioning on the IV is not necessary. In this example, $U$ is an unobserved API variable that poses a threat to identification. If it is related to the outcome (through the dashed arrow), the exclusion restriction is violated. An observed NCI ($NC$) that affects the API variable ($U$) can be used to evaluate the presence of the problematic association by implementing an unconditional independence test for the null hypothesis $H_0:NC\indep Y$. Panel~B shows that a variable causally affecting the IV can serve as a valid NCI. The variable $U$ is an unobserved API variable that poses a threat to identification---if it is related to the outcome (through the dashed arrow), outcome independence is violated. The square around $Z$ symbolizes the conditioning on the IV. An observed NCI ($NC$) that affects the IV ($Z$) can be used to evaluate the presence of the suspected association by testing $NC\indep Y|Z$. Specifically, if $NC\cancel{\indep} Y|Z$, then the path $NC\rightarrow Z \leftarrow U \rightarrow Y$ exists and hence $U \cancel{\indep} Y | Z$.}
\end{figure}

\clearpage

%====================================== TABLES =======================================%

\begin{landscape}
\addtab{tab_intro_examples}{Examples of Negative Control Tests for IV in Economics}{}
\end{landscape}

\begin{landscape}
\addtab{tab_survey}{Current Use of Falsification Tests for IV in Economics}{The table shows the results of our survey of highly cited articles employing instrumental variable (IV) designs published in leading economics journals from 2013 to 2023. The sample includes all articles from this period in the Review of Economic Studies (REStud), American Economic Review (AER), Journal of Political Economy (JPE), Quarterly Journal of Economics (QJE), and Econometrica (ECMA) that used IV designs and had significant citation counts on Google Scholar (over 300 citations for papers until 2020, and over 100 for those published after 2020). We examined these papers for their use of falsification tests. Column~(2) shows the proportion of papers that employ any falsification test. Columns (3)--(8) report the fraction of papers that implemented different types of falsification tests out of all papers that implemented any falsification test. 
Columns (3)--(5) categorize the tests into negative control outcome, negative control instrument, and other types of falsification tests, respectively. The fractions do not sum up to one, as some papers employed multiple test types. Columns (6)--(8) provide the share of papers using simple linear designs where the reduced form equation is modified by: replacing the outcome with an NCO, replacing the IV with an NCI, and adding the NCI (while still conditioning on the original IV). Column~(9) reports the median number of negative control variables used. \ref{AppSec:survey} provides additional details on the survey construction.}
\end{landscape}

\addtab{tab_applications_nco}{Illustrative Applications of Negative Control Outcome Tests}{
This table presents $p$-values from different NCO tests using data from \citet{autor2013china} and \citet{deming2014using}.
Column (1) replicates one of the original falsification analyses, in which \citeauthor{autor2013china} replaced the outcome with the NCO in the same 2SLS specification as their main analysis (ibid., Table 2, Part II). \citeauthor{deming2014using} conducted no falsification tests.  Columns 2--6 report $p$-values obtained from additional tests that include the same controls as in the most exhaustive specification of the original analyses. Column (2) reports a single test using one NCO, where the outcome is replaced with the NCO in the reduced form regression. For \citeauthor{autor2013china} the single NCO is the lagged outcome (in 1970), which is the same NCO reported in Column (1); for \citeauthor{deming2014using} this is lagged test scores (2002).
%Columns 3 to 6 use a broader set of NCOs identified from the data (52 for \citeauthor{autor2013china}; 37 for \citeauthor{deming2014using}).
Column (3) presents a Bonferroni-corrected $p$-value for multiple tests using all the NCOs, with the same specification as Column (2). 
Column (4) uses an F-test (Model~\eqref{Eq:LinReg}) with all NCOs jointly.
Columns (5) and (6) use GAM tests with linear and smoothed controls, respectively (Model~\eqref{Eq:GAM}).
}
\vspace{.5in}
\addtab{tab_applications_nci}{Illustrative Applications of Negative Control Instrument Tests}{
This table presents $p$-values from different NCI tests using data from \citet{nunn2014us} and \citet{ashraf2013out}, applying their original sets of NCIs (three and ten NCIs, respectively). Column (1) shows a single linear NCI test that, inappropriately, does not condition on the IV. The NCI with the lowest $p$-value is shown (grape production for \citeauthor{nunn2014us} and distance from Mexico City for \citeauthor{ashraf2013out}).
Columns (2)--(4) condition on the IV: Column (2) implements a proper linear NCI test (Model~\eqref{Eq:pseudo-IV-conditional}) using the same NCI as Column (1); Column (3) applies Bonferroni correction for multiple linear NCI tests; Column~(4) uses an F-test for all NCIs jointly.

}

\clearpage

\appendix

\begin{appendices}

\renewcommand{\thesection}{Appendix \Alph{section}} % Ensures sections are labeled A, B, C...
\renewcommand{\theHsection}{appendix.\Alph{section}} % Avoids hyperlink issues

\setlength{\topmargin}{0in}        % Ensures top margin is 1 inch
\setlength{\textheight}{8.5in}       % Content height = 9 inches
\setlength{\headheight}{14pt}      % Standard header height
\setlength{\headsep}{20pt}         % Space between header and text
\setlength{\footskip}{20pt}        % Space between text and footer

\makeatletter
\setlength{\textwidth}{6.5in}   % Ensure 1-inch left/right margins
\setlength{\oddsidemargin}{0in}  % Compensate for LaTeX’s default offset
\setlength{\evensidemargin}{0in} % Make even/odd pages equal
\makeatother

\renewcommand{\baselinestretch}{1.2}
\normalsize
\setlength{\parskip}{6pt}
\setlength{\parindent}{0pt}

\pagenumbering{arabic}% resets `page` counter to 1
\renewcommand*{\thepage}{\roman{page}}

\renewcommand{\figurename}{}
\renewcommand{\thefigure}{Appendix Figure~\Alph{section}\arabic{figure}}
\setcounter{figure}{0}

\renewcommand{\tablename}{}
\renewcommand{\thetable}{Appendix Table~A\arabic{table}}
\setcounter{table}{0}

\renewcommand{\theequation}{A\arabic{equation}}  % note that AER style says numbrering should be within appendix sections.
\setcounter{equation}{0}

\setcounter{ass}{0} 
\renewcommand{\theass}{A\arabic{ass}}

\setcounter{dfn}{0} 
\renewcommand{\thedfn}{A\arabic{dfn}}
\setcounter{thm}{0} 

\renewcommand{\thethm}{A\arabic{thm}}

\begin{center}
    {\LARGE \textbf{Online Supplementary Appendices}} \\[1em]
    {\LARGE \textbf{Negative Control Falsification Tests for\\ Instrumental Variable Designs}} \\[1em]
    {Oren Danieli (orendanieli@tauex.tau.ac.il), Daniel Nevo, Itai Walk, Bar Weinstein, Dan~Zeltzer} \\[1em]
    
    \today
\end{center}

\section{Additional Tables and Figures}

\addtab[1]{tab_applications_intro}{Papers Used to Illustrate Applications of Negative Control Tests}{This table provides contextual details on the papers we use for our negative control application examples in \ref{tab_applications_nco} and \ref{tab_applications_nci}. Columns (1)--(3) specify which outcome, treatment, and IV were used in these papers, respectively. Column (4) depicts the negative controls used in the original analysis, and Column (5) presents the number of negative controls used. Column (6) presents the number of negative controls used in the analysis for \ref{tab_applications_nco} and \ref{tab_applications_nci}, including the original negative controls, as well as additional valid negative controls we found in the original data.}
\begin{figure}[htbp]
\caption{Steps for Implementing Negative Control Tests}
\label{fig:steps}
\begin{mdframed}
\scalebox{0.85}{\begin{minipage}{1.0\linewidth} 

\begin{enumerate}
\item \textbf{Articulate Potential Threats to IV Design:}
Characterize potential threats as \emph{alternative paths} between the instrument (\(Z\)) and the outcome (\(Y\)) through an \emph{alternative path variable} (\(U\)):
\begin{enumerate}
\item Which assumption is potentially violated, outcome independence (Assumption~\ref{Ass:Indep}) or exclusion restriction (Assumption~\ref{Ass:ExclRestriction})? 
\item What is the type of the alternative path variable (\(U\))?
\begin{itemize}
\item Alternative Path Outcome (APO):  
\(U\) is associated with the outcome (\(Y\)), and the researchers are concerned about a potential association with the IV (\(Z\)) (Definition~\ref{dfn:apo_var1}).  
\item Alternative Path Instrument (API):  
\(U\) is associated with the IV (\(Z\)), and the researchers are concerned about a potential association with the outcome (\(Y\)) (Definition~\ref{dfn:api_var1}).  
\end{itemize}  
        \end{enumerate} 
         *Both outcome independence and exclusion restrictions can be violated through either APO or API variables. See Section~\ref{subsec:threat} for definitions and \ref{fig:dags_intro} for illustrations. 
    \item \textbf{Survey Available Data to Identify \emph{Negative Control} Variables:}
        \begin{itemize}
            \item Negative Control Outcomes (NCOs) proxy for APO variables.  
            \item Negative Control Instruments (NCIs) proxy for API variables.  
        \end{itemize}  
        Section~\ref{SubSubSec:NCassumptions} discuss the assumptions these proxies need to satisfy.  
        Examples are discussed in Section~\ref{SubSec:chooneNC}.  

    \item \textbf{Select Negative Control Test Specification:}
        \begin{itemize}
            \item NCO tests: assess the independence between the IV and NCO variables. Regress IV on NCOs and controls to also test Rich Covariates (Assumption \ref{ass:RC}).
            \item NCI tests: evaluate the independence between the NCI variables and the outcome, typically conditional on the IV. 
        \end{itemize}
    For IVs valid conditional on control variables, negative control tests should also condition on these controls. 

    \item \textbf{Interpret Results:}
        
 If the null is rejected:
  \begin{itemize}
      \item Investigate violations of IV design or functional form assumptions separately.  
      \item For multiple NCs, identify which are most predictive of the IV or outcome to diagnose specific threats.  
   \end{itemize}
If the null is not rejected:
  \begin{itemize}
      \item Consider insufficient power or weak negative controls.  
    \item Note that untested alternative paths may still exist.  
\end{itemize}
        
\end{enumerate}
\end{minipage}}

\end{mdframed}
\end{figure}

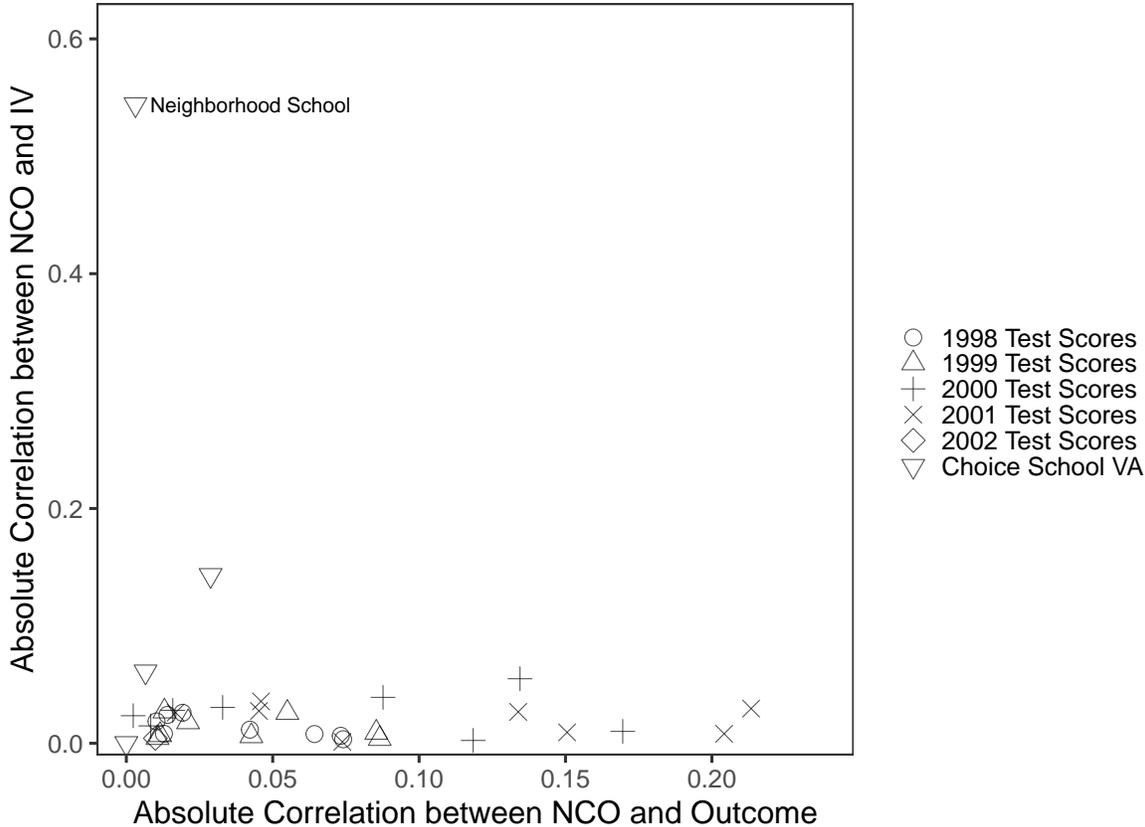
\addfig{corplot_lott_VA.pdf}{Correlations of NCOs with the IV and the Outcome in \cite{deming2014using}}{This figure shows a scatter plot of the absolute value of the correlation of different NCOs with the IV (on the y-axis) and the outcome (on the x-axis) in \citet{deming2014using}. The NCOs, the IV, and the outcome were first residualized by regressing them on all control variables and lottery fixed-effects. Each observation is one NCO.  For presentation purposes, NCOs are grouped into categories denoted by marker shape. The different year markers refer to groups of students' test scores from that year. The VA marker denotes the value added of the schools listed by students as their 1st, 2nd, and 3rd submitted preferences, as well as their neighborhood school's VA (labeled Neighborhood School). See Section~\ref{SubSec:DemingDetails} for details.}

\clearpage

\section{Details on Survey of Common Practices}
\label{AppSec:survey}
\textbf{Sample Construction.}
We used Google Scholar in November 2023 to assemble the list of relevant papers. We searched the terms ``instrumental variable,'' ``instrument,'' ``2SLS,'' and ``IV.'' We restricted the sample to articles with over 300 citations or, if published after 2020, over 100 citations. We examined all articles satisfying these criteria, published in five top-ranked economics journals: Review of Economic Studies, American Economic Review, Journal of Political Economy, Quarterly Journal of Economics, and Econometrica. Overall, our survey includes 140 papers.

We then searched the papers for strings related to falsification testing. This included ``falsification,'' ``negative control,'' ``balance,'' ``balancing,'' ``valid,'' and ``validity.'' Papers that did not include any of these strings were marked as not having any falsification test. We manually coded the type of falsification test for papers that included one of these strings. The results are summarized in \ref{tab_survey} and discussed in Section~\ref{sec:survey}.

\textbf{Other Falsification Tests.}
As discussed in Section~\ref{sec:survey}, we categorized all falsification tests used in surveyed papers into NCO tests, NCI tests, and other falsification tests. Other falsification tests include the following: negative control tests in non-IV settings, which examine only the first or second stage in a 2SLS estimation; ``placebo population'' analyses \citep{eggers2021placebo,glymour2012credible,keele2019falsification}, which involve repeating the analysis using a different population where the IV is not expected to affect the outcome; validating that the results are robust to including additional control variables; and using an over-identification test when more than one IV is available.

\beginQuotable[DAGtheory]
\section{Formalization of IV Negative Control Tests Using DAG Theory}
\label{AppSec:DAG}

In this section, we present an alternative formalization of the theory of negative controls for IV designs using the language of DAGs. We first summarize fundamental concepts from DAG theory.\footnote{This is not an exhaustive overview of DAG use for causality, but only the elements necessary for our purposes. See \cite{pearl2009causality} for detailed presentation and theory.} The DAGs we use intuitively throughout the paper can be formalized using the theory reviewed in this section. 

\subsection{Background DAG Theory}
A directed graph is a set of nodes and directed edges. A \textit{directed path} on a graph between two nodes, $X_1$ and $X_2$, is a sequence of edges, such that the first edge starts at $X_1$, each edge starts at the arrowhead of the former edge, and the last edge ends at $X_2$. A directed \emph{acyclic} graph (DAG) contains no cycles; namely, there are no directed paths from a node to itself. While DAGs can be used to represent assumptions on joint probability functions without notions of causality, here we interpret all DAGs as causal, such that the arrows represent a causal relationship. Consider a DAG denoted by $G$. The set of \textit{parents} of node $X_j$, denoted by $PA_j$, is the set of all nodes with direct arrows into $X_j$. The \textit{descendants} of $X_j$ are all nodes with a directed path of any number of edges (including a single edge) from $X_j$ to those variables; these are variables causally affected by $X_j$, directly or indirectly.

DAGs can be used to encode conditional independence assumptions on the joint distribution of variables represented by nodes in the DAG. Each DAG represents an infinite number of probability functions sharing the same conditional independence structure. 
The joint distribution of all variables in the DAG is the product of the conditional probability function of each variable $X_j$ given its parents $PA_j$.\footnote{This assumption is called the Markov property.} Formally, for $M$ variables $X_1,...X_M$, this factorization is 
$$
P(x_1,...,x_M)=\prod_{j=1}^MP(x_j|PA_j),
$$ 
where lower cases are realizations. Any probability function $P$ that admits to the above factorization is said to be compatible with a DAG $G$.

A key result of DAG theory is the translation of the structure represented by a DAG into conditional independence conditions. This translation relies on the concept of $d$-separation, a graphical condition. 

\begin{dfn}[$d$-separation \citep{pearl2009causality}]
\label{dfn:d-sep}
    A path $p$ from node $X_1$ to node $X_2$ in graph $G$ is said to be blocked by a set of nodes $\mathcal{A}$ if and only if
    \begin{enumerate}
        \item The path $p$ contains a chain $X_1 \rightarrow A \rightarrow X_2$ or a fork $X_1 \leftarrow A \rightarrow X_2$ such that the middle node $A$ is in $\mathcal{A}$, or
\item The path $p$ contains a collider structure $X_1 \rightarrow B \leftarrow X_2$ such that the middle node, the collider $B$, is not in $\mathcal{A}$ and
such that no descendant of $B$ is in $\mathcal{A}$.
    \end{enumerate}
The set $\mathcal{A}$ is said to $d$-separate node $X_1$ from node $X_2$ if and only if $\mathcal{A}$ blocks every path between $X_1$ and $X_2$. In this case, we write $$(X_1\indep X_2| \mathcal{A})_G.$$
\end{dfn}

Building on the above structure, the following theorem states the direct implication of $d$-separation, a graphical condition, on conditional independence, a probabilistic statement. 
\begin{thm}[\citealt{pearl2009causality}]\label{thm:DAG2CI}
    Let $X_1,X_2$ be two variables that are $d$-separated by a set of variables $\mathcal{A}$ in graph $G$, $(X_1\indep X_2| \mathcal{A})_G$. Then,  
    $$
    X_1\indep X_2|\mathcal{A}
    $$
    in all probability functions compatible with $G$.\footnote{Often $X_1\indep X_2|\mathcal{A}$ is written as $(X_1\indep X_2|\mathcal{A})_P$. We omit the $P$ subscript for simplicity. }
\end{thm}

IVs can also be defined using a graphical criterion \citep[][Definition 7.4.1]{pearl2009causality}. Consider a DAG $G$ with nodes $Z,X,Y$ (and possibly additional nodes, as in the DAGs in \ref{fig:dags_intro}). We follow \citet{pearl2009causality} by using $G_{\overline{X}}$ to denote the version of the DAG $G$ with all arrows entering $X$ removed.

\begin{dfn}
\label{dfn:IV_DAG}
    A variable $Z$ is an IV for treatment $X$ on outcome $Y$ if 
    \begin{enumerate}
        \item $(Z \indep Y)_{G_{\overline{X}}} $
        \item $(Z \cancel{\indep} X)_G$.
    \end{enumerate}
\end{dfn}

The first condition corresponds to the condition stated in \eqref{Eq:IVexogeneity}. It implies that there is no alternative path between the IV and the outcome other than through the treatment. The second condition corresponds to the IV relevance assumption. The negative control falsification tests for IVs we discuss in this paper focus on the first condition. 

\subsection{Negative Controls for IVs Using DAGs}

We are now ready to present the theory of negative controls for IVs using DAGs. To this end, we begin by providing the definitions of APO and API variables and those of NCOs and NCIs. Then, we provide proofs of Theorems \ref{Thm:NCOdep} and \ref{Thm:NCIdep} using DAG theory.

We start with the graphical DAG-based definition of an APO variable.
\begin{dfn}[Alternative path outcome variable\label{dfn:apo_var_dag}]
A random variable $U$ is an APO variable if the following two conditions hold.
 \begin{enumerate}
\item  \textit{Latent IV validity}. $(Z \indep Y | U)_{G_{\overline{X}}}$.
\item \textit{Path indication}. If $(Z \indep Y)_{G_{\overline{X}}}$ then $(Z \indep U)_G$.
\end{enumerate}
\end{dfn}
This definition is similar to the definition of APO variables using potential outcomes (Definition~\ref{dfn:apo_var1}). Latent IV validity implies that $Z$ and $Y$ are $d$-separated by $U$ (excluding the path through $X$). Path indication implies that if the IV is valid, then there is no unblocked path between $Z$ and $U$. 

Similarly, we can define graphically an API variable.
\begin{dfn}[Alternative path instrument variable\label{dfn:api_var_dag}]
A random variable $U$ is an API variable if the following two conditions hold.
 \begin{enumerate}
\item  \textit{Latent IV validity}. $(Z \indep Y | U)_{G_{\overline{X}}}$.
\item \textit{Path indication}. If $(Z \indep Y)_{G_{\overline{X}}}$ then $(U \indep Y | Z)_G$.
\end{enumerate}
\end{dfn}

Turning to negative control variables, 
the definitions of NCO and NCI are similar to the definitions using potential outcomes (Definitions~\ref{dfn:NCO_var} and \ref{dfn:NCI_var}). 
\begin{dfn}\label{dfn:NCO_var_DAG}
A random variable $NC$ is an NCO if there exists an APO variable $U$ such that the following two conditions hold.
\begin{enumerate}
\item The NCO assumption. 
 $(NC \indep Z |\ U)_G$.
 \item $U$-comparability. $(NC \cancel{\indep} U)_G$.
\end{enumerate}
\end{dfn}

\begin{dfn}\label{dfn:NCI_var_DAG}
A random variable $NC$ is an NCI if there exists an API variable $U$ such that the following two conditions hold.
\begin{enumerate}
\item The NCI assumption.  $(NC \indep Y |\ Z,U)_G$.
\item $U$-comparability. $(NC\cancel{\indep}U|Z)_G$. 
\end{enumerate}

\end{dfn}

We now provide a proof of Theorem \ref{Thm:NCOdep} under the above DAG definitions. Following the condition of Theorem \ref{Thm:NCOdep}, the NCO test finds that  $NC \cancel{\indep} Z$.  By Definition~\ref{dfn:NCO_var_DAG} we have that $(NC \indep Z |\ U)_G$. From the test we know that $NC \cancel{\indep} Z$, which implies $(NC \cancel{\indep} Z)_G$ by the contrapositive of Theorem~\ref{thm:DAG2CI}. Because $(NC \cancel{\indep} Z)_G$, there is at least one open (unblocked) path between $Z$ and $NC$. However, because $(NC \indep Z | U)_G$, this path or paths are blocked by $U$. By Definition \ref{dfn:d-sep}, this means that $U$ is either in the middle of a chain or a fork on the open path between $Z$ and $NC$. Thus, there is an unblocked path between $Z$ and $U$, i.e., $(Z \cancel{\indep} U)_G$ which, by path indication, implies that $(Z \cancel{\indep} Y)_{G_{\overline{X}}}$, violating the first IV condition in Definition~\ref{dfn:IV_DAG}. 

We turn to the proof of Theorem \ref{Thm:NCIdep} under the DAG definitions. Following the condition of Theorem~\ref{Thm:NCIdep},  the NCI test finds that  $NC \cancel{\indep} Y | Z$.  By Definition \ref{dfn:NCI_var_DAG} we have that $(NC \indep Y |\ Z, U)_G$. From the test, we know that $NC \cancel{\indep} Y|Z$, which implies $(NC \cancel{\indep} Y|Z)_G$ by the contrapositive of Theorem \ref{thm:DAG2CI}. Because $(NC \cancel{\indep} Y | Z)_G$, there is at least one open path between $NC$ and $Y$, which is not blocked by $Z$. However, because $(NC \indep Y | U,Z)_G$, this path or paths are blocked by $U$. By Definition \ref{dfn:d-sep}, this means that $U$ is either in the middle of a chain or a fork on the open path between $Y$ and $NC$. Thus, there is an unblocked path between $U$ and $Y$ not blocked by $Z$, i.e., $(U \cancel{\indep} Y|Z)_G$ which by path indication implies that $(Z \cancel{\indep} Y)_{G_{\overline{X}}}$, violating the first IV condition in Definition \ref{dfn:IV_DAG}.

\endQuotable{}

\section{Additional Theory and Proofs}
\label{AppSec:proofs}

Throughout, we let $P(\cdot | \cdot)$ be the conditional probability or density function. As a shorthand, we leave the random variables to be understood from the arguments of $P$. 
For example, if $Y(x)$ is discrete, $P[y(x)| u]$ is a shorthand for $\Pr[Y(x)=y(x)| U=u]$.

 \subsection*{Auxiliary Lemmas}
\begin{lem}[Lemma 4.3 in \cite{dawid1979conditional}]
\label{Lem:indep1}
Let $A,B,D,Q$ be four random variables. If $A\indep B| D,Q$ and $B\indep Q| D$ then $A\indep B | D$.
\end{lem}

\begin{proof}
Because $A \indep B | D,Q$, it follows that for all $a,b,d,q$, we have that
\begin{align}
\begin{split}
\label{AppEq:lem1aux}P(a,b|d,q)&=P(a|d,q)P(b|d,q)\\
&=P(a|d,q)P(b|d),
\end{split}
\end{align}
where the last line follows from $B\indep Q| D$.
Now,
\begin{equation*}
P(a,b|d)=\int P(a,b|d,q)P(q|d)dq = \left[\int P(a|d,q)P(q|d)dq\right]P(b|d)=P(a|d)P(b|d),
\end{equation*}
where the second equality is by \eqref{AppEq:lem1aux}.
\end{proof}
This lemma is also known as the contraction axiom of conditional independence \citep{pearl1986graphoids}. The following lemma is a direct result of Lemma  \ref{Lem:indep1}.

\begin{lem}
\label{Lem:indep2}
Let $A,B,D,Q$ be four random variables. If $A\indep B|D,Q$ and $A\cancel{\indep}B|D$ then $A\cancel{\indep}Q|D$
and $B\cancel{\indep}Q|D$.
\end{lem}

\begin{proof}
Assume by way of contradiction that $B\indep Q|D$. Therefore, by Lemma \ref{Lem:indep1}, because  $A\indep B| D,Q$ it follows that $A\indep B| D$,
which contradicts the assumption.  A similar contradiction is received by assuming $A\indep Q| D$. 
\end{proof}

\subsection{Negative Controls for IV Designs Under General Definitions}

This section presents the proofs of the theoretical results from Section \ref{Sec:theory}. We prove versions of the results that are more general in three different ways. First, we discuss IV designs that include control variables. Second, we provide more general definitions for APO and API variables that accommodate multiple threats of which Definitions \ref{dfn:apo_var1} and \ref{dfn:api_var1} are special cases. Third, we provide more general definitions of NCO and NCI  (under weaker NCO and NCI assumptions, respectively). 

We start by presenting the outcome independence and exclusion restriction assumptions when controls are included.

\begin{ass}[Outcome independence]
\label{Ass:IndepWithControls}
$Z \indep Y(z,x)|C$ for all possible $z,x$ values.
\end{ass} 
\begin{ass}[Exclusion restriction]
\label{Ass:ExclRestrictionWithControls}
$\Pr(Y(z,x)= Y(z',x)=Y(x)|C=c)=1$ for all possible $z,z',x,c$ values.
\end{ass}
Similar to the case without controls, outcome independence and exclusion restriction together yield $Z\indep Y(x)|\ C$.

\subsubsection{Alternative Path Variables with Multiple Threats and Controls}
\label{AppSubSubSec:NCOproofs}
In some applications, multiple potential alternative paths can exist between the IV and the outcome. \ref{appSec:multiple_APO_Var} presents an example of two distinct variables that affect the outcome and could potentially affect the IV as well and thus may violate outcome independence. To accommodate the possibility of multiple violations of the IV assumptions, we extend Definition~\ref{dfn:apo_var1} and \ref{dfn:api_var1}. We introduce a random variable $V$, which represents other potential threats in addition to the threat posed by the alternative path variable $U$. We also include control variables $C$ to accommodate cases where the IV design is assumed to be valid only when controls are included. 

The more general definition for APO variables reads as follows. 
\begin{dfn}[Alternative path outcome variable with multiple threats and controls]\label{dfn:app-apo}
A random variable $U$ is an APO variable conditional on a set of controls $C$ if there exists a random variable $V$ such that the following conditions hold.
\begin{enumerate}
\item  \textit{Latent IV validity}. $Z \indep Y(x) | C, U, V.$ 
 \item \textit{Path indication.} If $Z \indep Y(x)| C,V$ then $Z \indep U | C, V$.
\item \textit{Direct IV link}. If $Z\indep U | C,V$ then $Z \indep U | C$.
\item \textit{V-validity}. If  $Z{\indep}Y(x)|C$  then $Z{\indep}  Y(x) | C, V$.
\end{enumerate}
\end{dfn}

\noindent Under this definition, latent IV validity states that the IV is valid conditionally not only on the APO variable $U$, but also on the additional threat(s) $V$, and the controls $C$.

In contrast to Definition~\ref{dfn:apo_var1}, this more general version of latent IV validity also holds for $U$ even if $V$ is the actual threat to the identification and $U$ is only an imperfect proxy for it. Therefore, to maintain the same interpretation of an APO variable as posing a threat to IV validity, we replace the condition of path indication from Definition~\ref{dfn:apo_var1} and include two additional conditions. Together, Conditions 2--4 of Definition \ref{dfn:app-apo} yield Condition 2 of Definition \ref{dfn:apo_var1} (conditional on the controls $C$). However, the three separate conditions ensure that the APO variable is part of the threat itself and not a proxy. Specifically, path indication and direct IV link each rule out a different type of proxy; see \ref{AppSubSec:VioloationProxies} and \ref{AppSubSec:VioloationProxiesPathIndication} for counterexamples.
The final property, V-validity, is a more technical requirement for the variable $V$ that ensures $V$ represents other threats. It states that an  IV design that satisfies outcome independence and exclusion restriction conditional on the controls remains valid conditional on $V$. See \ref{AppSubSec:ViolationVvalidityColider} for a counterexample.  
If there are no additional threats other than through $U$, and no control variables, Definition \ref{dfn:app-apo} is equivalent to Definition \ref{dfn:apo_var1}.

We similarly extend Definition \ref{dfn:api_var1} to allow for additional threats and to include controls $C$.

\begin{dfn}[Alternative path instrument variable with multiple threats and controls\label{dfn:api_var2}]
A random variable $U$ is an API variable conditional on a set of controls $C$ if there exists a random variable $V$ such that the following conditions hold.
 \begin{enumerate}
\item  \textit{Latent IV validity}. $Z \indep Y(x) | C,U,V$.
\item \textit{Path indication}. If $Z \indep Y(x)|C,V$ then $U \indep Y|Z, C,V$.
\item \textit{Direct outcome link}. If  $U{\indep}Y|Z, C,V$  then $U{\indep}  Y | Z, C$.
\item \textit{V-validity}. If  $Z{\indep}Y(x)|C$  then $Z{\indep}  Y(x) | C,V$.
\end{enumerate}
\end{dfn}
Condition 1 is the same as in Definition \ref{dfn:app-apo}. Conditions 2--4 together imply a version of Condition 2 from Definition \ref{dfn:api_var1} that includes controls. Similar to the theory of APO variables, the condition is decomposed into three independent conditions to exclude proxies that are not themselves part of an alternative path and maintain that $V$ is indeed a part of such a threat (similar to Definition \ref{dfn:app-apo}).

\subsubsection{General Definition of Negative Control Variables}

We adapt the definition of NCOs as follows.

\begin{dfn}[Negative control outcome with controls]
\label{dfn:NCO_var_controls}
A random variable $NC$ is an NCO if there exists an APO variable $U$ such that the following two conditions hold.
\begin{enumerate}
    \item The NCO assumption.   If $Z\indep U|C$ then $NC \indep Z |\ U,C.$
 \item $U$-comparability. $NC \cancel{\indep} U|C$.
\end{enumerate}
\end{dfn} 

Even without controls ($C=\emptyset$), this definition is more general than Definition \ref{dfn:NCO_var}, as it allows for NCOs that were previously excluded.  To see this, note that in the case without controls, if $Z\cancel{\indep} U$  (and so the design is invalid), the NCO assumption in Definition \ref{dfn:NCO_var_controls} allows for an association $NC \cancel{\indep} Z |\ U$. In this case, the NCO would still be informative about the validity of the IV design since the association between the NCO and the IV, given the APO variable exists only if the design is invalid. \ref{AppSubSec:NCOaffectZ} provides an example of such an NCO. 
%In this example, there is a concern that an IV is not randomly assigned as argued. Instead, it might be affected by the APO variable, as well as observed lagged variables. The researcher can then test if the IV is valid by examining whether the IV is associated with these lagged variables. While the lagged variables might directly affect the IV, they are still valid NCOs as they are not directly associated with the IV ($Z \indep NC |U$) when the design is potentially valid ($Z\indep U$).  
Every variable that satisfies the NCO assumption in Definition \ref{dfn:NCO_var} trivially satisfies this less restrictive definition.

Next, we generalize the definition of an NCI to include controls and allow for direct associations with the outcome if the IV design is invalid.

\begin{dfn}[Negative control instrument with controls]
\label{dfn:NCI_var_with_controls}
A random variable $NC$ is an NCI if there exists an API variable $U$ such that the following two conditions hold.

\begin{enumerate}
    \item The NC assumption.   If $U \indep Y|Z,C$ then $NC \indep Y|Z,C,U.$
  \label{nci_assumption_controls}  
 \item $U$-comparability. $NC \cancel{\indep} U|Z,C$.
\end{enumerate}
\end{dfn}

\subsubsection{Negative Control Tests with Controls}

We are now ready to state the more general version of Theorem \ref{Thm:NCOdep} and present its proof. This theorem also covers the case without controls by letting $C$ be degenerate.

\beginQuotable[ThmNCOControls]
\begin{thm}\label{Thm:NCOdepControls}
Assume that a random variable $NC$ is an NCO with respect to controls $C$ (Definition~\ref{dfn:NCO_var_controls}). If $NC\cancel{\indep}Z|C$, then either outcome independence or exclusion restriction is violated. That is, the IV design is invalid.
\end{thm} 
\endQuotable{}
\begin{proof}
We begin by showing that $NC\cancel{\indep}Z|C$ implies that $Z\cancel{\indep}U|C$. Else, if $Z\indep U|C$ then by the NCO assumption (see Definition~\ref{dfn:NCO_var_controls}) $NC \indep Z |\ U,C$. Based on Lemma \ref{Lem:indep2}, $NC \indep Z |\ U,C$ and $NC\cancel{\indep}Z|C$ imply that $Z\cancel{\indep}U|C$, a contradiction. 

Next, from direct IV link if follows that $Z\cancel{\indep}U|C$ implies $Z\cancel{\indep}U|V,C$. Then, by path indication, we get that $Z\cancel{\indep}Y(x)|V,C$. Finally, by V-validity, we have that $Z\cancel{\indep}Y(x)|C$.

However, outcome independence (Assumption \ref{Ass:IndepWithControls}) and exclusion restriction (Assumption \ref{Ass:ExclRestrictionWithControls}) together imply that $Z\cancel{\indep}Y(x)|C$. Therefore, one of these assumptions must be violated. 
\end{proof}

Next, we state the more general version of Theorem \ref{Thm:NCIdep} and present its proof. This theorem also covers the case without controls by letting $C$ be degenerate.

\begin{thm}%[Negative Control Instrument Test]
\label{Thm:NCIdepControls}
Assume that a random variable $NC$ is an NCI with respect to controls $C$  (Definition~\ref{dfn:NCI_var_with_controls}). If $NC\cancel{\indep}Y|Z,C$, then either outcome independence or exclusion restriction is violated. That is, the IV design is invalid.
\end{thm}
\begin{proof}
We divide the proof into two cases with respect to the API variable $U$ for which the NCI assumption holds for NC.
%\begin{enumerate}[label=(\alph*)]
%     \item $U \cancel{{\indep}} Y|Z,C$,
%     \item $U {\indep} Y|Z,C$.
% \end{enumerate}

First, assume that $U \cancel{{\indep}} Y|Z,C$. In this case, from direct outcome link it follows that $U \cancel{\indep} Y|Z, C,V$, and therefore by path indication, $Z \cancel{\indep} Y(x)|C,V$. Therefore, by $V$-validity, we have that $Z \cancel{\indep} Y(x) |C$. Hence, either outcome independence or the exclusion restriction do not hold. 

We now turn to the other case, where  $U {\indep} Y|Z,C$. 
By the NCI assumption (Definition \ref{dfn:NCI_var_with_controls}), we have that $NC{\indep} Y|Z,C,U$, and by the condition of the theorem, we have that $NC \cancel{\indep} Y | Z,C$. Therefore, by  Lemma \ref{Lem:indep2} (with $D= \{Z,C\}, Q=U$), we have that  $U \cancel{\indep} Y|Z,C$, which contradicts the assumption $U {\indep} Y|Z,C$. 

\end{proof}

We now turn to state and prove a version of Theorem \ref{thm:NCIdepUnconditional}, conditional on controls $C$.

\begin{thm}
\label{Thm:NCIdepUnconditionalControls}
Assume that a random variable $NC$ is an NCI with respect to controls $C$ (Definition~\ref{dfn:NCI_var_with_controls}). If, in addition, $NC\indep Z|C$, then if  $NC\cancel{\indep}Y|C$, then $Z\cancel{\indep}Y(x)|C$. 

\end{thm}
\begin{proof}
Assume by way of contradiction that $Z \indep Y(x)| C$ holds. Because  $NC$ is an NCI, it follows from Theorem \ref{Thm:NCIdepControls} that $NC \indep Y|Z,C$. Additionally, based on the assumption that $NC \indep Z|C$, Lemma \ref{Lem:indep1} implies that $NC \indep Y|C$, which contradicts the premise. 
\end{proof}

\subsubsection{Control Variables and Functional Form} \label{AppSubSec:corollaries}
\textbf{Proof of Corollary \ref{Coro:NCO_RC}}
\begin{proof}

The  minimized expression  can be written as  
$$
\bbE[Z-b'_CC - b_{NC}NC]^2 = \bbE\big[Z-\bbE[Z|C,NC]\big]^2 +\bbE\big[\bbE[Z|C,NC]-b'_CC - b_{NC}NC\big]^2
$$ 
plus a term equaling zero because $\bbE[Z-\bbE[Z|C,NC]]$  is zero by the law of total expectation.
Since $\bbE[Z-\bbE[Z|C,NC]]^2$ does not depend on $b_C,b_{NC}$, we can write 
$$
\gamma = \argmin_{b_c,b_{NC}}\bbE\big[\bbE[Z|C,NC]-b'_CC - b_{NC}NC\big]^2.
$$
If outcome independence, exclusion restriction, and rich covariates are assumed, \eqref{Eq:NCnull2SLS} holds and $\bbE[Z|C,NC]=\gamma'_CC$. Hence, we can further write 
$$\gamma = \argmin_{b_c,b_{NC}}\bbE[\gamma'_CC-b'_CC - b_{NC}NC]^2.$$
The values that minimize this  nonnegative expression are $b'_C = \gamma'_C$ and $b_{NC} = 0$ and so the OLS population-level coefficient is $\gamma' = (\gamma'_C, 0).$ If $\gamma_{NC} \neq 0$, it must be that (\ref{Eq:NCnull2SLS}) does not hold. Therefore, either outcome independence, exclusion restriction, or rich covariates is violated.

\end{proof}

\textbf{Proof of Corollary \ref{Coro:pseudo-outcome}}
\begin{proof}
Let $\widetilde{Z}=Z-C'\bbE[C C']^{-1}\bbE[CZ]$ and $\widetilde{NC}=NC-C'\bbE[C C']^{-1}\bbE[CNC]$ be the residuals from the linear regressions of $Z$ and $NC$ on $C$, respectively. By the Frisch–Waugh–Lovell theorem, we can write 
$\beta_Z$ as %using $M_C=I-C(C^T C)^{-1} C^T$ as
\begin{equation*}
        \beta_Z = \frac{COV(\widetilde{Z},\widetilde{NC})}{Var(\widetilde{Z})}.
\end{equation*}
If $\beta_Z \neq 0$ then it must be that $COV(\widetilde{Z}, \widetilde{NC})  \neq 0$. 
Define $(\gamma'_C,\gamma_{NC})$ to be the population-level solution of the reverse OLS, with $Z$ as the dependent variable (as in Corollary \ref{Coro:NCO_RC}).
Again, by the Frisch–Waugh–Lovell theorem, we can write $\gamma_{NC}$ as
\begin{equation*}
\gamma_{NC} =\frac{COV(\widetilde{Z},\widetilde{NC})}{Var(\widetilde{NC})}.
\end{equation*}
Since $COV(\widetilde{Z}, \widetilde{NC})  \neq 0$ it follows that $\gamma_{NC} \neq 0$ as well. 
By Corollary~\ref{Coro:NCO_RC}, we have that either outcome independence, exclusion restriction, or rich covariates does not hold.

\end{proof}

\textbf{Proof of Corollary \ref{Coro:pseudo-IV}}
\begin{proof}
Similar to the proof of Corollary  \ref{Coro:NCO_RC}, we write the equivalent minimization problem as 
$$
\theta =  \argmin_{b_Z,b_c,b_{NC}}\bbE\big[\bbE[Y|Z,C,NC]-b_ZZ-b'_CC - b_{NC}NC\big]^2.
$$ 
If outcome independence, exclusion restriction, and CSRF are assumed,  \eqref{Eq:NCInullCSRF} holds and 
$$\theta =  \argmin_{b_Z,b_c,b_{NC}}\bbE[\theta_ZZ + \theta'_CC-b_ZZ-b'_CC - b_{NC}NC]^2.$$ 
The values that minimize this expression are $b_Z = \theta_Z, b'_C = \theta'_C$ and $b_{NC} = 0$ and so the OLS population-level coefficient is $\theta' = (\theta_Z,\theta'_C, 0).$ If $\theta_{NC} \neq 0$, it must be that (\ref{Eq:NCInullCSRF}) does not hold. Therefore either outcome independence, exclusion restriction, or CSRF is violated.

\end{proof}

\section{Examples and Counterexamples}
\label{AppSec:Examples}

\subsection{Non-Causal APO Variable}
\label{AppSubSec:NonCausalAPV}

 \begin{figure}[htbp!]
 \begin{center}
 	\begin{tikzpicture}
  	\node (1) {$U_1$};
 	\node [right =   1.5cm of 1](2) {$U_2$};
 	\node [above =   1cm of 2](3) {$NC$};
    \node [right =   2cm of 2](4) {$Z$};
    \node [right =  1.5cm of 4](5) {$X$};
    \node [right =  1.5cm of 5](6) {$Y$};
    \node [above left =  1cm of 5](7) {$W$};
    \draw[Arrow] (1) -- (2);
    \draw[Arrow] (2) -- (3);
    \draw[DashedRedArrow] (2) -- (4);
    \draw[Arrow] (4) -- (5);
    \draw[Arrow] (5) -- (6);
    \draw[Arrow] (7) -- (5);
    \draw[Arrow] (7) -- (6);
    \draw [Arrow] (1) to [bend right=25] (6);
 	\end{tikzpicture}
 	 \end{center}
\caption{An illustration of causal and non-causal APO variables}
\label{AppFig:NonCausalAPV}
 \end{figure}

In the example presented in \ref{AppFig:NonCausalAPV}, $U_2$ is a valid APO variable satisfying path indication without having a direct causal effect on $Y$ (beyond possible effect through the IV). Consider a scenario where  $Z$ represents supposedly quasi-random teacher assignment, $X$ is the value added of the actual teacher, and $Y$ is student test scores. The variable $W$ represents some unobserved student characteristic (e.g., their parents' involvement), which affects test score and also correlates with teacher assignments, as some students switch classrooms after the original assignment. 

In this example, $U_1$ denotes unobserved student ability that directly affects test scores, while $U_2$ represents detailed unobserved previous exam scores. The dashed arrow indicates the possibility that principals may assign students to teachers based on these detailed scores (e.g., students with low math scores are assigned to specific teachers). The detailed past exam scores satisfy path indication despite not directly affecting future test scores. A path exists between previous detailed scores ($U_2$) and current scores ($Y$) because both are affected by ability ($U_1$). 

The variable $NC$ represents aggregated previous test scores (e.g., average past scores in math together with other subjects). In this setting, $NC$ is an NCO, with $U_2$ as an APO variable. An association between the IV and the aggregated lagged test scores would indicate an alternative path from the IV to the outcome. 
The presence of this path violates outcome independence, as students with different abilities would sort into different teachers based on previous math scores. 

Note that while $U_1$ is also an APO variable, $NC$ is a valid NCO only with respect to $U_2$, not with respect to $U_1$ alone. This is because, conditional on $U_1$, there is still a correlation between the NCO and the IV ($NC \cancel{\indep} Z | U_1$). That is, teacher assignment is not conditionally independent of aggregated test scores.

\subsection{Heterogeneity-Based Violation of Path Indication}
\label{AppSubSec:ViolationPathHetero}

\begin{figure}[htb!]
\begin{center}
     \begin{subfigure}[b]{0.49\textwidth}
	\begin{tikzpicture}
	\node (1) {$W$};
	\node [above left = 1cm of 1](2) {$Z$};
	\node [right =   1.5cm of 2](3) {$X$};
	\node [right =   2cm of 3](4) {$Y$};
	\node [above right =  1.5cm of 2](5) {$U$};
	\node [right =  1cm of 5](6) {$N$};
	\draw[Arrow] (1) -- (3);
	\draw[Arrow] (1) -- (4);
	\draw[Arrow] (2) -- (3);
	\draw[Arrow] (3) -- (4);
	\draw[Arrow] (5) -- (6);
    \draw[Arrow] (5) -- (4);
	\end{tikzpicture}
	        \caption{Females}
	\end{subfigure}
	 \begin{subfigure}[b]{0.49\textwidth}
	\begin{tikzpicture}
	\node (1) {$W$};
	\node [above left = 1cm of 1](2) {$Z$};
	\node [right =   1.5cm of 2](3) {$X$};
	\node [right =   2cm of 3](4) {$Y$};
	\node [above right =  1.5cm of 2](5) {$U$};
	\node [right =  1cm of 5](6) {$N$};
	\draw[Arrow] (1) -- (3);
	\draw[Arrow] (1) -- (4);
	\draw[Arrow] (2) -- (3);
	\draw[Arrow] (3) -- (4);
	\draw[Arrow] (5) -- (6);
    \draw[Arrow] (2) -- (5);
	\end{tikzpicture}
	        \caption{Males}
	\end{subfigure}
	\end{center}
	\caption{Violation of path indication (Definition \ref{dfn:apo_var1}) due to heterogeneity}
	\label{Fig:ViolationVvalid}
\end{figure}

\ref{Fig:ViolationVvalid} describes a random variable $U$ associated with both the IV and the potential outcome $Y(x)$. However, $U$ does not qualify as an APO variable because it does not satisfy path indication. The IV affects $U$ for males but not for females, while $U$ affects $Y$ for females but not for males. Path indication is not satisfied because even though the IV satisfied outcome independence and exclusion restriction assumptions ($Z \indep Y(x)$), it remains associated with $U$ ($Z \cancel{\indep} U$). 

The random variable $N$ is a proxy for $U$. With sufficient sample size, we would detect $N \cancel{\indep} Z$ (due to the effect among males), but this test result does not indicate that the IV is invalid since $U$ is not an APO variable. 

\subsection{Violation of Path Indication: Multivariate Variable}
\label{AppSubSec:ViolationPathMultiDim}

 \begin{figure}[ht!]
 \begin{center}
 	\begin{tikzpicture}
  	\node (1) {$U_1$};
 	\node [below =   0.2cm of 1](2) {$U_2$};
\node [above right =   0.75cm of 1](3) {$N$};
\node [right =   2cm of 1](4) {$Z$};
\node [right =  1.5cm of 4](5) {$X$};
\node [right =  1.5cm of 5](6) {$Y$};
\node [above left =  1cm of 5](7) {$W$};
\node [above = 0.5cm of 1, color = blue] (8) {$U$};
\draw[Arrow] (1) -- (4);
\draw[Arrow] (1) -- (3);
\draw[Arrow] (4) -- (5);
\draw[Arrow] (5) -- (6);
\draw[Arrow] (7) -- (5);
\draw[Arrow] (7) -- (6);
\draw[DashedRedArrow] (2) -- (6);    \draw[line width=0.5mm, color= blue] (0,-0.5) circle [x radius=1, y radius=0.5, rotate=90];
 	\end{tikzpicture}
 	 \end{center}
\caption{Violation of path indication (Definition \ref{dfn:apo_var1}) when $U$ has multiple components}
\label{AppFig:vectorNotAPV}
 \end{figure}

\ref{AppFig:vectorNotAPV} presents an example where $U=(U_1,U_2)$ is a bivariate vector of independent variables. Assume $Z$ represents teacher assignment claimed to be quasi-random, $X$ is the actual teacher's value added, and $Y$ is test scores. The variable $W$ represents some unobserved student characteristic (e.g., their parents' involvement), which also correlates with teacher assignments, as some students switch classrooms after the original assignment. Let $U_1$ represent having basketball as a hobby and further assume that it is correlated with the IV; for example, one teacher also coaches basketball, so students who list basketball as a hobby are more likely to be assigned to her. However, as seen in \ref{AppFig:vectorNotAPV}, a basketball hobby is independent of test scores ($U_1 \indep Y(x)$). 
Let $U_2$ represent having math as a hobby. Students reporting math as a hobby tend to perform better in exams and are randomly distributed between teachers ($U_2 \indep Z$). The basketball and math hobbies are independent ($U_1 \indep U_2$). Finally, assume that $N$ is participation in an extracurricular basketball program, serving as a proxy for $U$ (specifically for $U_1$).

Although the vector $U$ is associated with both the IV and the outcome, it does not qualify as an APO variable because it does not satisfy path indication. The IV satisfies outcome independence and the exclusion restriction assumptions ($Z \indep Y(x)$) despite being correlated with the list of hobbies ($Z \cancel{\indep} U$).  
Therefore, $N$ is not a proper NCO. Even though $N \indep Z |U$, it is still not an NCO because $U$ is not an APO variable.

\subsection{Potential Alternative Path Variables and Association with the Treatment}
\label{AppSubSec:APvarAndTreatmnet}
\beginQuotable[AppFigAPvarAssociatedTreatment]
\begin{figure}[htp!]
\bigskip
\centering
\subfloat[][$NC\cancel{\indep}Z$ implies violation of IV validity.  $U$ is an APO variable. $NC$ is an NCI.]{
\begin{tikzpicture}
\node (1) {$Z$};
\node [above right = 0.75cm and 0.25cm of 1](2) {$W$};
\node [right =   1cm of 1](3) {$X$};
\node [right =   1cm of 3](4) {$Y$};
\node [left = 1cm of 1](5) {\color{red} $U$};
\node [below  = 0.5cm of 1](6) {\color{blue}$NC$};
\draw[Arrow] (1) -- (3);
\draw[Arrow] (2) -- (3);
\draw[Arrow] (2) -- (4);
\draw[Arrow] (3) -- (4);
\draw[DashedRedArrow] (5) to  (1);
\draw[Arrow] (5) to [out=25, in=145] (3);  % Adjusted arrow
\draw[Arrow] (5) to [out=-25, in=-135] (4);
\draw[Arrow] (5) -- (6);
\end{tikzpicture}\hspace{2em}}
\hspace{3em}
 \subfloat[][$N\cancel{\indep}Y|Z$ does not necessarily implies violation of IV validity. $U$ is not an API variable. $N$ is not an NCO.]{
\begin{tikzpicture}
\node (1) {$Z$};
\node [above right = 0.75cm and 0.25cm of 1](2) {$W$};
\node [right =   1cm of 1](3) {$X$};
\node [right =   1cm of 3](4) {$Y$};
\node [left = 1cm of 1](5) {\color{red} $U$};
\node [below  = 0.5cm of 1](6) {\color{blue}$N$};
\draw[Arrow] (1) -- (3);
\draw[Arrow] (2) -- (3);
\draw[Arrow] (2) -- (4);
\draw[Arrow] (3) -- (4);
\draw[Arrow] (5) to  (1);
\draw[Arrow] (5) to [out=25, in=145] (3);
\draw[DashedRedArrow] (5) to [out=-25, in=-135] (4);
\draw[Arrow] (5) -- (6);
\end{tikzpicture}\hspace{2em}}

\caption{Direct effect of the potential alternative path variables on the treatment}
\label{AppFig:APvarAssociatedTreatment}
\begin{minipage}{\textwidth}
\medskip
\endQuotable{}
\footnotesize
\end{minipage}
\end{figure}
Path indication for API variables (Definition \ref{dfn:api_var1}) implies that conditional on the IV ($Z$), an API variable ($U$) cannot be associated with the treatment ($X$). By contrast, there is no such requirement for an APO variable. \ref{AppFig:APvarAssociatedTreatment} illustrates these points. In Panel A, $U$ is a valid APO variable, and the arrow $U\rightarrow X$ is allowed: both latent IV validity and path indication hold. Therefore, $NC$ is a valid NCO, and an association between $NC$ and $Z$ implies that the dashed arrow between $U$ and $Z$ exists and indicates that the IV is invalid. Conversely, in Panel B, $U$ is not a valid API variable. Path indication is violated because $U \cancel{\indep} Y|Z$ does not imply $Z\indep Y(x)$. Therefore, $N$ is not an NCI, and $N\cancel{\indep}Y|Z$ does not necessarily imply that the IV is invalid. Intuitively, $N \cancel{\indep} Y|Z$, even if the IV design is valid, due to the association between $U$ and $Y$ through $X$. Conditioning on $X$ would not solve this problem because $X$ is a collider. Therefore, $N \cancel{\indep} Y|Z,X$ due to the path $N\leftarrow U \rightarrow X \leftarrow W \rightarrow Y$ \citep{pearl2009causality}.

\subsection{Counterexample: A Vector of NCOs That is Not an NCO}
\label{AppSubSec:NCVecNotNC}

Let $R_1,R_2$ be two independent Bernoulli random variables with probabilities $\Pr(R_j=1)=p_j$ and let $p_1=p_2 =0.5$. Let $U$ be another Bernoulli random variable, independent of $(R_1, R_2)$. Let $Z$ be the IV, and assume that  
$$
Z=(R_1 \oplus R_2)+ \theta U + \epsilon_Z,
$$ 
where $\oplus$ is the XOR operator. Assume that $ Y(x) = x + U+\epsilon_Y$, such that $U$ is an APO variable. The IV design is valid if $\theta = 0$. 

Now, assume that there are two observed negative controls $NC_i  = U \oplus R_i$ for $i = 1,2$. Both $NC_1$ and $NC_2$ are valid negative controls as they satisfy the assumption $NC_i \indep Z | U$. This is because for $i = 1,2$, $R_i \indep (R_1 \oplus R_2)$, and therefore  $Z \indep R_i | U$. However $(NC_1, NC_2)\cancel{\indep} Z|U$ because, conditional on $U$, $Z$ is associated with $NC_1 \oplus NC_2 = R_1 \oplus R_2$. Therefore, $(NC_1, NC_2)$ does not satisfy the NCO assumption. Indeed, even if the IV is valid, we could still have $Z \cancel{\indep} (NC_1, NC_2)$.

A small change in the data-generating process will break some of the independencies discussed above. For example, changing the value of $p_1$ to something different from 0.5 would imply that $R_2 \cancel{\indep} (R_1 \oplus R_2)$. In that case, $NC_2 \cancel{\indep} Z$ and $NC_2$ would no longer satisfy the NCO assumption.  
\addtocounter{figure}{1}

\subsection{Multiple Threats}\label{appSec:multiple_APO_Var}

\begin{figure}[htb!]
\begin{tikzpicture}
\node (1) {$Z$};
\node [right =   1cm of 1](2) {$X$};
\node [above left = 0.5cm of 2](3) {$W$};
\node [right =   1cm of 2](4) {$Y$};
\node [above left = 0.5cm and 1cm of 1](5) {\color{red} $U$};
\node [below left = 0.5cm and 1cm of 1](6) {\color{red} $V$};

\draw[Arrow] (1) -- (2);
\draw[Arrow] (3) -- (2);
\draw[Arrow] (3) -- (4);
\draw[Arrow] (2) -- (4);
\draw[DashedRedArrow] (5) to  (1);
\draw[DashedRedArrow] (6) to  (1);
\draw[Arrow] (5) to [out=25, in=120] (4.north);
\draw[Arrow] (6) to [out=-25, in=-120] (4.south);

\end{tikzpicture}\hspace{2em}
   \centering
    \caption{Multiple threats}
    \label{fig:multiple_APO_Var}
\end{figure}
\ref{fig:multiple_APO_Var} presents an example of the presence of multiple threats to the validity of the IV design. 
 In this case, the variable $U$ is an APO variable by Definition \ref{dfn:app-apo} (taking the control variables $C$ to be an empty set). In this figure, $V$ is an APO variable as well.

For example, assume that $Z$ is the teacher assignment, which is claimed to be quasi-random, $X$ is the value added of the actual teacher, and $Y$ is test scores. The variable $W$ represents some unobserved student characteristic (e.g., their parents' involvement), which also correlates with teacher assignments, as some students switch classrooms after the original assignment. Assume that $U$ is the student's unobserved ability. Assume also that $V$ is principal quality, which is also unobserved. Both $U$ and $V$ might affect the teacher allocation $Z$, which would generate an alternative path between the IV and the outcome. 

\subsection{Direct IV Link Rules out Proxies of $V$}
\label{AppSubSec:VioloationProxies}

\begin{figure}[ht!]
 \begin{center}
 	\begin{tikzpicture}
\node (1) {$V$};
\node [above right=0.75cm and 1cm of 1](2) {$Z$};
\node [right =  1.5cm of 2](3) {$X$};
\node [right =  1.5cm of 3](4) {$Y$};
\node [above left =  1cm of 3](5) {$W$};
\node [above =  1.5cm of 1](6) {$U_2$};
\node [left =  1.5cm of 1](7) {$U_1$};
\draw[DashedRedArrow] (1) -- (2);
     \draw[Arrow] (1) -- (6);
     \draw[Arrow] (2) -- (3);
     \draw[Arrow] (3) -- (4);
     \draw[Arrow] (5) -- (3);
     \draw[Arrow] (5) -- (4);
     \draw[Arrow] (7) -- (1);
     \draw[Arrow] (1) -- (4);
 	\end{tikzpicture}
 	 \end{center}
\caption{Violation of direct IV link (Definition \ref{dfn:app-apo})}
\label{AppFig:ViolationProxyDirectIVlink}
 \end{figure}

\ref{AppFig:ViolationProxyDirectIVlink} presents the potential violation of outcome independence through $V$, as well as two proxies for $V$, $U_1$, and $U_2$. Note that latent IV validity, as stated in Definition \ref{dfn:app-apo}, holds for either variable ($U_1$ or $U_2$) together with $V$, as the further conditioning on $U_1$ or $U_2$ does not invalidate the IV, conditional on $V$. 
Note also that for both variables, path indication holds because if $Z\indep Y(x)|V$ then $Z\indep U_1|V$ (or $Z\indep U_2|V$). 
However, condition 3 of Definition \ref{dfn:app-apo}, direct IV link, does not hold. If the IV design is invalid (the dashed line exists), then $Z\cancel{\indep} U_1$ while $Z\indep U_1| V$ (and similarly for $U_2$).
Intuitively, we rule out $U_1$ and $U_2$ as APO variables because they are only proxies for the variable creating the threat to IV validity.

\subsection{Path Indication Rules Out Proxies of $V$}
\label{AppSubSec:VioloationProxiesPathIndication}

\begin{figure}[ht!]
 \begin{center}
 	\begin{tikzpicture}
 	\node (2) {$Z$};
  	\node [below right=0.75cm and 1cm of 2](1) {$V$};
    \node [right =  1.5cm of 2](3) {$X$};
    \node [right =  1.5cm of 3](4) {$Y$};
    \node [above left =  1cm of 3](5) {$W$};
 %   \node [below =  1.5cm of 1](6) {$U_2$};
    \node [left =  1.5cm of 1](7) {$U$};
     \draw[DashedRedArrow] (2) -- (1);
 %    \draw[Arrow] (1) -- (6);
     \draw[Arrow] (2) -- (3);
     \draw[Arrow] (3) -- (4);
     \draw[Arrow] (5) -- (3);
     \draw[Arrow] (5) -- (4);
     \draw[Arrow] (7) -- (1);
     \draw[Arrow] (1) -- (4);

    %\draw[Arrow,dashed] (2) -- (4);
    
    %\draw [Arrow] (1) to [bend right=30] (6);
 	\end{tikzpicture}
 	 \end{center}
\caption{Violation of path indication (Definition \ref{dfn:app-apo})}
\label{AppFig:ViolationProxyPathIndication}
 \end{figure}
 
\ref{AppFig:ViolationProxyPathIndication} presents a potential violation of exclusion restriction through the variable $V$, as well as a proxy for $V$, labeled as $U$. While $V$ itself is an APO, its proxy $U$ is not. Note that latent IV validity holds for $U,V$ jointly, as the further conditioning on $U$ does not invalidate the IV design once we have conditioned on $V$. Note also that direct IV link holds because $Z\indep U$. However, $U$ is not an APO variable because it does not represent a threat to IV validity. Condition~2 of Definition~\ref{dfn:app-apo}, namely, path indication, does not hold. Specifically, if the IV design is invalid (the dashed line exists), $Z\cancel{\indep} U|V$ while $Z\indep Y(x)| V$. In the language of DAG terminology, $V$ is a \textit{collider} \citep{pearl2009causality}, and conditioning on it creates a dependence between $U$ and $Z$.

\subsection{Violation of V-validity}
\label{AppSubSec:ViolationVvalidityColider}

\begin{figure}[ht!]
 \begin{center}
 	\begin{tikzpicture}
  	\node (1) {$U$};
 	\node [above right=0.75cm and 1cm of 1](2) {$Z$};
    \node [right =  1.5cm of 2](3) {$X$};
    \node [right =  1.5cm of 3](4) {$Y$};
    \node [above left =  1cm of 3](5) {$W$};
    \node [right =  1.5cm of 4](6) {$V$};
%    \node [above = 1cm of 1](7) {$N$};
    \draw[Arrow] (1) -- (2);
    \draw[Arrow] (1) -- (6);
    \draw[Arrow] (2) -- (3);
    \draw[Arrow] (3) -- (4);
    \draw[Arrow] (5) -- (3);
    \draw[Arrow] (5) -- (4);
    \draw[Arrow] (4) -- (6);
%    \draw[Arrow] (1) -- (7);

    %\draw[Arrow,dashed] (2) -- (4);
    
    %\draw [Arrow] (1) to [bend right=30] (6);
 	\end{tikzpicture}
 	 \end{center}
\caption{Violation of $V$-validity}
\label{AppFig:ViolationVvalidity}
 \end{figure}
\ref{AppFig:ViolationVvalidity} presents a situation with no valid APO variable. We examine $U$ as a candidate APO variable and consider $V$ in the DAG as the potential $V$ in Definition \ref{dfn:app-apo}. We see that latent IV validity holds: while $Z \cancel{\indep} Y(x) | V$, because $V$ is a common effect (a collider) of $U$ and $Y$, controlling for $U$ in addition to $V$ blocks the flow of association \citep{pearl2009causality}, resulting in $Z \indep Y(x) | U, V$. Path indication holds because $Z \cancel{\indep} Y(x) | V$. Direct IV link also holds because of the effect of $U$ on $Z$. However, it is clear $U$ should not be an APO variable. An association between $Z$ and $U$ does not imply that the IV design is invalid. This is where $V-$validity comes to the rescue. The IV satisfies $Z \indep Y(x)$, but, as previously noted $Z \cancel{\indep}
Y(x)|V$, due to $V$ being a common effect of both variables.
In this case, no other alternative to $V$ exists to satisfy Definition \ref{dfn:app-apo}. Therefore, $U$ is not an APO variable.%, and the random variable $N$ is not an NCO.  

\subsection{NCO Potentially Affecting the IV} \label{AppSubSec:NCOaffectZ}

\begin{figure}[htbp!]
\bigskip
\centering
\subfloat[][$Z\indep U$]{
\begin{tikzpicture}
\node (1) {$Z$};
\node [above right = 0.75cm and 0.25cm of 1](2) {$W$};
\node [right =   1cm of 1](3) {$X$};
\node [right =   1cm of 3](4) {$Y$};
\node [left = 1cm of 1](5) {\color{red} $U$};
\node [above right = 0.75 and 0.2cm of 5](6) {\color{blue}$NC$};
\draw[Arrow] (1) -- (3);
\draw[Arrow] (2) -- (3);
\draw[Arrow] (2) -- (4);
\draw[Arrow] (3) -- (4);

\draw[Arrow] (5) to [out=-25, in=-135] (4);
\draw[Arrow] (5) -- (6);\end{tikzpicture}\hspace{2em}}
\hspace{3em}
 \subfloat[][$Z \cancel{\indep} U$ ]{

\begin{tikzpicture}
\node (1) {$Z$};
\node [above right = 0.75cm and 0.25cm of 1](2) {$W$};
\node [right =   1cm of 1](3) {$X$};
\node [right =   1cm of 3](4) {$Y$};
\node [left = 1cm of 1](6) {$U$};
\node [above right = 0.75 and 0.2cm of 6](7) {\color{blue}$NC$};
\draw[Arrow] (1) -- (3);
\draw[Arrow] (2) -- (3);
\draw[Arrow] (2) -- (4);
\draw[Arrow] (3) -- (4);
\draw[Arrow] (6) to [out=-25, in=-135] (4);
\draw[RedArrow] (6) -- (1);
\draw[Arrow] (6) -- (7);
\draw[RedArrow] (7) -- (1);
%\draw[Arrow] (5) to [out=-25, in=180] (6);
\end{tikzpicture}
\hspace{2em}}
\caption{Conditional dependence between an NCO and the IV}
\label{AppFig:directNC-ZLink}
\begin{minipage}{\textwidth}
\medskip
\end{minipage}
\end{figure}

\ref{AppFig:directNC-ZLink} presents a scenario in which if the IV is invalid, it could also be associated with the NCO, not through the APO variable. For concreteness, consider the case of studying the effect of teacher quality ($X$) on test scores ($Y$). The IV ($Z$) is claimed to be a random assignment of teachers. Unobserved ability ($U$) is the APO variable. In the case of random assignment, ability has no association with the IV (Panel A). However, there is a concern that the initial assignment was not random in practice. In Panel B, random assignment did not take place, and so other considerations could have impacted the IV, including unobserved ability $U$. Moreover, it is possible that proxies for unobserved ability, such as lagged test scores ($NC$), were used directly in the assignment process as well. In this case, $NC \cancel{\indep} Z | U$, and so the lagged outcome does not satisfy Definition \ref{dfn:NCO_var}. However, if the IV design is valid, and $Z \indep U$, the condition $NC \indep Z | U$ is satisfied. Hence, $NC$ is an NCO based on the broader Definition \ref{dfn:NCO_var_controls}, defined in \ref{AppSubSubSec:NCOproofs}. Indeed, in this case, if $NC \cancel{\indep} Z$, the design is invalid ($Z \cancel{\indep} Y(x)$).

\section{Details of Implementation of Negative Control Tests Using Data from Prior Studies}
\label{appendix_practice}

This section provides additional details for the analysis from Section~\ref{sec:practice}, which implements our proposed methods on IV designs used in prior studies. \ref{tab_applications_intro} summarizes information about the key variables in each study. We are grateful to the authors of these prior studies for publicly posting their data and code. In each case, we first used the publicly posted data to replicate the related original study's results (this step is not further discussed here). We then applied our additional negative control falsification tests. Replication code and data are included in the supplementary materials.

\subsection{Implementation Details for \citet{autor2013china}} \label{SubSec:ADHdetails}

\textbf{Sample Construction.} For this analysis, we use the original study's data from \citet[henceforth ADH]{autor2013china}, which is taken from the US Census. The unit of analysis is a commuting zone. The sample included 722 commuting zones.

\textbf{Main Variables.} For each commuting zone, we observe all variables from the original study's replication data and additional variables not used in the original study, some of which we use as NCOs in our current analysis. The treatment and IV are built as shift-share variables, weighting change in Chinese import by industry where weights are the local industry shares in the commuting zone. The treatment uses Chinese imports in the US and the IV uses Chinese imports in other developed countries to avoid endogeneity.
We focus on the analysis for the years 2000--2007. The treatment and IV are the shift-share difference in Chinese imports between the years 2007 and 2000. The control variables are the lagged year 2000 values. Note that ADH also used another version of the IV, measured between 1990--2000. We do not evaluate this version because it would not allow us to use the large set of variables from 1990 as NCOs.

\textbf{Original Falsification Tests.} ADH conducted falsification exercises to evaluate the concern that the decline in US manufacturing employment in commuting zones with high exposure to Chinese imports might have occurred for reasons unrelated to Chinese imports. They regress past changes in the manufacturing employment share on future changes in import exposure (See Columns (4)--(6) of Table 2 in ADH). This relationship was found to be significant only for 1970--1980, but not for 1980--1990 or 1970--1990. The significant specification yielded a coefficient with the opposite sign. We replicated this analysis and obtained a similar result. The $p$-value is reported in Column (1) of our \ref{tab_applications_nco}. This original exercise is similar in spirit to our proposed approach, although it uses the different negative controls separately and not jointly. It also uses a 2SLS specification for estimation, not the reduced form.

The remainder of this section discusses additional falsification tests that we performed using alternative negative control variables sourced from the original replication data. 

\textbf{Additional NCOs.} We use 52 NCOs in our falsification analysis. These include the NCOs that were originally used by ADH (lagged changes in manufacturing employment) and all variables measuring labor market conditions in 1990. In particular, we include the share of workers who were employed in manufacturing, employed in non-manufacturing, unemployed, and not in the labor force, separately for males, females, college educated, non-college educated, and for three different age groups; the share who received SSDI; average log weekly wages in manufacturing and in non-manufacturing; average household total income and average household wage; total population and size of the workforce; levels of transfers per capita for medical benefits, federal income assistance, unemployment benefits, TAA benefits, education/training assistance, SSA retirement benefits, SSA disability benefits, other assistance, and total individual transfers. 

\textbf{Implementation Details.} We use the same sampling weights used by ADH in the original study (\texttt{timepwt48}). We also follow ADH and cluster standard errors by states (\texttt{statefip}). 

In Column (1) of \ref{tab_applications_nco}, we use a single NCO that was used in the original analysis, namely the change in manufacturing employment between 1970--1980.
We replicated the ADH analysis, which regressed past outcomes (1970) on future treatments (years 1990 and 2000 averaged), instrumented by future IVs (see Column (4) of Table 2 in ADH). We report the $p$-value of the coefficient on the treatment. In Column (2), we perform a similar analysis by regressing the 1970 outcome on the year 2000 IV (e.g., reduced form), including the full set of 16 control variables (as in Column (6) of Table 3 in ADH).

\subsection{Implementation Details for \cite{deming2014using}}\label{SubSec:DemingDetails}

\textbf{Sample Construction.} We use the original study's data from a public school choice lottery in Charlotte-Mecklenburg, North Carolina. The unit of analysis is the individual student. The sample includes 2,343 students. 

\textbf{Main Variables.} We use Deming's VAM estimates from the mixed-effects specification, controlling for past test scores.\footnote{The original study included richer specifications (models 3--4 in the original study) that controlled for individual characteristics, which were not made publicly available due to privacy constraints.}
Based on the replication code, we can write the IV as \begin{equation}
IV_{i}=L_{i}VAM_{i}^{1}+\left(1-L_{i}\right)VAM_{i}^{N}\label{Eq:deming IV}
\end{equation}
where $L$ is the binary school lottery outcome, $VAM^{1}$ is the value added of the first-choice school, and $VAM^{N}$ is the value added of the default neighborhood school. These variables are included in the original study's replication data.

Control variables include lagged test scores from the year 2001--2002 as well as lottery fixed effects (i.e., a categorical variable for every choice of school ranking). Following \cite{deming2014using}, the test scores include the math and reading test scores in nominal, quadratic, and cubic values, and an indicator of missing values.

\textbf{NCOs.} The original study did not report any falsification tests. We perform falsification analysis using lagged test scores from earlier school years (1998--2001) that were included in the replication data but not included as controls in the study (see the control variables definition) and lagged outcome (\texttt{testz2002}; i.e., 2002 test scores). We also used the VAM of the three schools that the student applied to in the lottery and the neighborhood school's VAM. In total, we used $37$ NCOs.

\textbf{Implementation Details.} Following the original paper, all our analyses are unweighted. In the analysis with a single NCO, we replace the outcome with lagged test scores (from 2001--2002) in the reduced form.
In the F-test and multiple linear tests with Bonferroni correction, we perform a fixed-effect regression of the IV on the NCOs with the \texttt{lottery\_FE} variable. In the GAM models, fixed effects are accounted for by taking \texttt{lottery\_FE} as a categorical variable without a smooth term.

\textbf{Additional Analysis.} In \ref{corplot_lott_VA.pdf}, we show the correlation of each NCO with the outcome and the IV. Before calculating each correlation, we residualized the NCO and the IV or the outcome by the control variables. 

In an unreported analysis, we replicated the main 2SLS results using $L_i$, the raw lottery outcome, as an alternative IV. The point estimates remained statistically unchanged, although standard errors were larger. 

\subsection{Implementation Details for \citet{nunn2014us}}
\label{SubSec:NunnQianDetails}

\textbf{Sample Construction.} We use the study data, which consists of annual panel data of 125 non-OECD countries over 36 years. The sample includes 4,572 observations.

\textbf{Main Variables.} The IV of the study is the US wheat production from the previous year. We limit our analysis to the main outcome variable of the study, which is the intrastate conflict indicator. We utilize the extended set of $238$ control variables (as in the ``baseline specification'' in Table~2 of \citet{nunn2014us}). 

\textbf{NCIs.} As in the original study, we used a set of ten NCIs. The NCIs are the lagged US production of various products that are not sent as foreign aid. 

\textbf{Original Falsification Tests.} \citet{nunn2014us} performed a falsification test (Table~5 in \citeauthor{nunn2014us}) with the aforementioned NCIs by estimating the reduced form equation
\begin{equation*}
    Y_i = NCI^j_i + IV_i + C_i + \epsilon_i
\end{equation*} for each of the ten $NCI^j$ and the ``baseline specification'' of the control variables.

\textbf{Implementation Details.} In all analyses, we follow \citet{nunn2014us} and cluster standard errors by country.

\textbf{Additional Analysis.} We can also implement a GAM model with linear controls. This test rejects the null hypothesis. The rejection is driven at least in part by a violation of the unnecessary CSRF Assumption (Assumption \ref{ass:linear_RF}). To test the functional form, we implement Ramsey's RESET test for misspecification with quadratic and cubic fitted values for the reduced form equation. This test results in a $p$-value lower than $1\%$, implying a misspecification. However, the large number of control variables does not allow for estimating a GAM model with smooth controls as well or for including interactions of the control variables. Therefore, we cannot assess IV validity separately. 

\subsection{Implementation Details for \citet{ashraf2013out}}
\label{SubSec:AshrafGalorDetails}

\textbf{Sample Construction.} The study data consists of a sample of 145 countries.

\textbf{Main Variables.} The outcome of the study is the historical population density, which is defined as the log population density in 1500 CE. The main IV is the migratory distance from Addis Ababa. We use the same set of four control variables included in the study.

\textbf{NCIs.} We use the same three NCIs as in the original study, which are the migratory distances from London, Tokyo, and Mexico City. 

\textbf{Implementation Details.} We follow \citet{ashraf2013out} and include a quadratic polynomial for both the IV and the NCIs. 
\clearpage\end{appendices}

\end{document}